\documentclass[a4paper,12pt]{preprint}
\usepackage[margin=1.3in]{geometry}  
\usepackage{times}
\usepackage{hyperref}
\usepackage{breakurl}
\usepackage{comment}
\usepackage{microtype}

\usepackage{graphicx}              
\usepackage{amsmath}
\usepackage{amscd}               
\usepackage{amsfonts} 
\usepackage{amssymb}             
\usepackage{amsthm}                
\usepackage{mathrsfs}
\usepackage{amsbsy}
\usepackage{mhequ}
\usepackage{tikz}
\usepackage{wasysym}
\usepackage{centernot}

\numberwithin{equation}{section}

\newtheorem{thm}{Theorem}[section]
\newtheorem{defn}[thm]{Definition}
\newtheorem{lem}[thm]{Lemma}
\newtheorem{prop}[thm]{Proposition}

\newtheorem{assumption}[thm]{Assumption}
\DeclareMathOperator{\id}{id}
\def\Wick#1{\,\colon\!\! #1 \!\colon}

\makeatletter
\def\th@newremark{\th@remark\thm@headfont{\bfseries}}
\makeatletter

\def\eps{\epsilon}
\theoremstyle{newremark}
\newtheorem{rmk}[thm]{Remark}

\definecolor{darkgreen}{rgb}{0.1,0.7,0.1}
\definecolor{darkred}{rgb}{0.7,0.1,0.1}
\newcommand{\bB}{\mathcal{B}}
\newcommand{\cC}{\mathcal{C}}
\newcommand{\dD}{\mathcal{D}}
\newcommand{\eE}{\mathcal{E}}
\newcommand{\fF}{\mathcal{F}}
\newcommand{\gG}{\mathcal{G}}

\newcommand{\iI}{\mathcal{I}}

\newcommand{\lL}{\mathcal{L}}
\newcommand{\mM}{\mathcal{M}}
\newcommand{\nN}{\mathcal{N}}
\newcommand{\oO}{\mathcal{O}}
\newcommand{\pP}{\mathcal{P}}
\newcommand{\qQ}{\mathcal{Q}}
\newcommand{\rR}{\mathcal{R}}

\newcommand{\tT}{\mathcal{T}}
\newcommand{\uU}{\mathcal{U}}
\newcommand{\vV}{\mathcal{V}}
\newcommand{\wW}{\mathcal{W}}
\newcommand{\xX}{\mathcal{X}}

\newcommand{\fC}{\mathfrak{C}}

\newcommand{\fM}{\mathfrak{M}}
\newcommand{\fR}{\mathfrak{R}}

\newcommand{\fK}{\mathfrak{K}}

\newcommand{\CA}{\mathcal{A}}
\newcommand{\CE}{\mathcal{E}}
\newcommand{\CP}{\mathcal{P}}

\newcommand{\CX}{\mathcal{X}}
\def\E{\mathbf{E}}

\def\eps{\epsilon}

\def\powerset{\mathscr{P}}
\def\Cum{\mathbf{E}_c}

\def\one{\mathbf{1}}

\def\emptyset{\mathop{\centernot\ocircle}}

\newcommand{\lfl}{\left\lfloor }  
\newcommand{\rfl}{\right\rfloor} 
\newcommand{\R}{\mathbf{R}}
\newcommand{\T}{\mathbf{T}}

\newcommand{\1}{\mathbf{1}}

\newcommand{\sE}{\mathscr{E}}

\newcommand{\sJ}{\mathscr{J}}
\newcommand{\sL}{\mathscr{L}}
\newcommand{\sM}{\mathscr{M}}

\def\Wck{{\scriptscriptstyle\mathrm{Wick}}}

\newcommand{\wM}{M^{\Wck}}

\newcommand{\wDelta}{\Delta^{\Wck}}
\newcommand{\wPi}{\Pi^{\Wck}}
\newcommand{\wf}{f^{\Wck}}

\newcommand{\ha}{\widehat{a}}

\newcommand{\heE}{\widehat{\eE}}

\newcommand{\hPi}{\widehat{\Pi}}

\colorlet{symbols}{blue!90!black}
\colorlet{testcolor}{green!60!black}

\def\${|\!|\!|}

\usetikzlibrary{shapes.misc}
\usetikzlibrary{shapes.symbols}
\usetikzlibrary{snakes}
\usetikzlibrary{decorations}
\usetikzlibrary{decorations.markings}

\def\drawx{\draw[-,solid] (-3pt,-3pt) -- (3pt,3pt);\draw[-,solid] (-3pt,3pt) -- (3pt,-3pt);}
\tikzset{
	root/.style={circle,fill=testcolor,inner sep=0pt, minimum size=2mm},
	dot/.style={circle,fill=black,inner sep=0pt, minimum size=1mm},
	var/.style={circle,fill=black!10,draw=black,inner sep=0pt, minimum size=
		2mm},
	bvar/.style={circle,fill=black!15,draw=white,inner sep=0pt, minimum size=
		8mm},
	dotred/.style={circle,fill=black!50,inner sep=0pt, minimum size=2mm},
	generic/.style={semithick,shorten >=1pt,shorten <=1pt},
	dist/.style={ultra thick,draw=testcolor,shorten >=1pt,shorten <=1pt},
	testfcn/.style={ultra thick,testcolor,shorten >=1pt,shorten <=1pt,<-},
	testfcnx/.style={ultra thick,testcolor,shorten >=1pt,shorten <=1pt,<-,
		postaction={decorate,decoration={markings,mark=at position 0.6 with {\drawx}}}},
	kepsilon/.style={semithick,shorten >=1pt,shorten <=1pt,densely dashed,->},
	kprimex/.style={semithick,shorten >=1pt,shorten <=1pt,densely dashed,->,
		postaction={decorate,decoration={markings,mark=at position 0.4 with {\drawx}}}},
	kernel/.style={semithick,shorten >=1pt,shorten <=1pt,->},
	multx/.style={shorten >=1pt,shorten <=1pt,
		postaction={decorate,decoration={markings,mark=at position 0.5 with {\drawx}}}},
	kernelx/.style={semithick,shorten >=1pt,shorten <=1pt,->,
		postaction={decorate,decoration={markings,mark=at position 0.4 with {\drawx}}}},
	kernel1/.style={->,semithick,shorten >=1pt,shorten <=1pt,postaction={decorate,decoration={markings,mark=at position 0.45 with {\draw[-] (0,-0.1) -- (0,0.1);}}}},
	kernel2/.style={->,semithick,shorten >=1pt,shorten <=1pt,postaction={decorate,decoration={markings,mark=at position 0.45 with {\draw[-] (0.05,-0.1) -- (0.05,0.1);\draw[-] (-0.05,-0.1) -- (-0.05,0.1);}}}},
	kernelBig/.style={semithick,shorten >=1pt,shorten <=1pt,decorate, decoration={zigzag,amplitude=1.5pt,segment length = 3pt,pre length=2pt,post length=2pt}},
	gepsilon/.style={dotted,semithick,shorten >=1pt,shorten <=1pt},
	renorm/.style={shape=circle,fill=white,inner sep=1pt},
	labl/.style={shape=rectangle,fill=white,inner sep=1pt},
	xi/.style={circle,fill=symbols!10,draw=symbols,inner sep=0pt,minimum size=1.2mm},
	xix/.style={crosscircle,fill=symbols!10,draw=symbols,inner sep=0pt,minimum size=1.2mm},
	xib/.style={circle,fill=symbols!10,draw=symbols,inner sep=0pt,minimum size=1.6mm},
	xibx/.style={crosscircle,fill=symbols!10,draw=symbols,inner sep=0pt,minimum size=1.6mm},
	not/.style={circle,fill=symbols,draw=symbols,inner sep=0pt,minimum size=0.5mm},
	>=stealth,
}

\makeatletter
\def\DeclareSymbol#1#2#3{\expandafter\gdef\csname MH@symb@#1\endcsname{\tikz[baseline=#2,scale=0.15,draw=symbols]{#3}}\expandafter\gdef\csname MH@symb@#1s\endcsname{\scalebox{0.7}{\tikz[baseline=#2,scale=0.15,draw=symbols]{#3}}}}
\def\<#1>{\csname MH@symb@#1\endcsname}
\makeatother

\DeclareSymbol{Xi22}{0.5}{\draw (0,0) node[xi] {} -- (-1,1) node[not] {} -- (0,2) node[xi] {}; }

\def\scal#1{\langle#1\rangle}

\begin{document}

\title{Weak universality of dynamical $\Phi^4_3$: non-Gaussian noise}
\author{Hao Shen$^1$ and Weijun Xu$^2$}
\institute{Columbia University, US, \email{pkushenhao@gmail.com}
\and University of Warwick, UK, \email{weijun.xu@warwick.ac.uk}}

\maketitle

\begin{abstract}
We consider a class of continuous phase coexistence models in three spatial dimensions. The fluctuations are driven by symmetric stationary random fields with sufficient integrability and mixing conditions, but not necessarily Gaussian. We show that, in the weakly nonlinear regime, if the external potential is a symmetric polynomial and a certain average of it exhibits pitchfork bifurcation, then these models all rescale to $\Phi^4_3$ near their critical point. 
\end{abstract}

\section{Introduction} \label{sec:intro}

The dynamical $\Phi^4_d$ model is formally given by
\begin{equation} \label{eq:phi43}
\partial_{t} \Phi = \Delta \Phi - \lambda \Phi^{3} + \xi, 
\end{equation}
where $\xi$ is the space-time white noise in $d$ spatial dimensions, namely the Gaussian random field with covariance formally given by
\begin{align*}
\E \xi(s,x) \xi(t,y) = \delta (s-t) \delta^{(d)} (x-y). 
\end{align*}
This equation was initially derived from the stochastic quantisation of the Euclidean quantum field theory \cite{PW}. It is also believed to be the scaling limit for Kac-Ising model near critical temperature (conjectured in \cite{GLP99}). Note that formally, one could rescale $\Phi$ to turn the coefficient in front of the cubic term to be $1$, but we still retain that $\lambda$ here in order to simplify the scaling later. 

In dimension $d=1$, the equation is well-posed, and it was shown in \cite{BPRS93} that the dynamic Ising model on the real line does rescale to it. 
However, as soon as $d \geq 2$, the equation becomes ill-posed in the sense that the ``solution" $\Phi$ is distribution valued so that the non-linearity $\Phi^{3}$ does not make sense. In $d=2$, the solution could be constructed either via Dirichlet forms (\cite{AR91}) or via Wick ordering (\cite{DPD03}, by turning the term $\Phi^{3}$ into the Wick product $\Wick{\Phi^{3}}$ with respect to the Gaussian structure induced by the linearised equation). Recently, Mourrat and Weber (\cite{Ising2d}) showed that $2$D Ising model with Glauber dynamics and Kac interaction does converge to the Wick ordered solution constructed in \cite{DPD03}. 

The three dimensional case is much more involved as it exhibits further logarithmic divergence beyond Wick ordering. It became amenable to analysis only very recently, either via the theory of regularity structures (\cite{Hai14a}), or via controlled para-products (\cite{GIP12, CC13}), or Wilson's renormalisation group approach (\cite{Antti}). In all these cases, the solution is obtained as the limit of smooth solutions $\Phi_{\epsilon}$ to the equations
\begin{equation} \label{eq:phi43_approx}
\partial_{t} \Phi_{\epsilon} = \Delta \Phi_{\epsilon} - \lambda \Phi_{\epsilon}^{3} + \xi_{\epsilon} + C_{\epsilon} \Phi_{\epsilon}, 
\end{equation}
where $\xi_{\epsilon}$ is a regularisation of the noise $\xi$ at scale $\epsilon$, and $C_{\epsilon} \rightarrow +\infty$ could be chosen such that the limit is \textit{independent} of the regularisation. For $d \geq 4$, one does not expect to get any nontrivial limit (\cite{Aizenman, Frohlich}). 

The three dimensional equation is also believed to be the universal model for phase coexistence near criticality. In the recent work \cite{HaiXu}, Hairer and the second author studied phase coexistence models in three spatial dimensions of the type
\begin{equation} \label{eq:micro_model}
\partial_{t} u = \Delta u - \epsilon V_{\theta}'(u) + \zeta, 
\end{equation}
where $\zeta$ is a space-time Gaussian random field with correlation length $1$, and $V_{\theta}$ is a polynomial potential whose coefficients depends smoothly on a parameter $\theta$. The main result of \cite{HaiXu}, built on analogous results for the KPZ equation in \cite{HQ15}, is that if $V$ is symmetric and its certain average exhibits a pitchfork bifurcation near the origin, then there exists a choice of
\begin{align*}
\theta = a \epsilon \log \epsilon + \oO(\epsilon)
\end{align*}
such that the rescaled process
\begin{equation} \label{eq:rescaled_solution}
u_{\epsilon}(t,x) := \epsilon^{-\frac{1}{2}} u(t/\epsilon^{2}, x/\epsilon), 
\end{equation}
converge to the solution of the $\Phi^4_3$ equation, interpreted as the limit in \eqref{eq:phi43_approx}, and the coefficient $\lambda$ of the cubic term in the limit is a linear combination of all coefficients $a_{j}(0)$'s of the polynomial $V_{0}$. We denote this limit as the $\Phi^4_3(\lambda)$ family (where the term ``family" comes from the fact that one could perturb the finite part of the infinite renormalisation constant, and end up with a one-parameter family of limits). The main purpose of this article is to show that when $\zeta$ is a symmetric, stationary noise with strong mixing but not necessarily Gaussian, then the macroscopic process $u_{\epsilon}$ still converges to the $\Phi^4_3$ family with a proper choice of $\theta$ (depending on $\epsilon$). 

To appreciate the difficulties with non-Gaussian noise, we first note that the macroscopic process $u_{\epsilon}$ satisfies the equation
\begin{equation} \label{eq:rescaled_equation}
\partial_{t} u_{\epsilon} = \Delta u_{\epsilon} - \epsilon^{-\frac{3}{2}} V_{\theta}'(\epsilon^{\frac{1}{2}} u_{\epsilon}) + \zeta_{\epsilon}, 
\end{equation}
where
\begin{equation} \label{eq:rescaled_noise}
\zeta_{\epsilon} (t,x) = \epsilon^{-\frac{5}{2}} \zeta (t/\epsilon^{2}, x/\epsilon)
\end{equation}
is an approximation to the space-time white noise at scale $\epsilon$. Since the right hand side of \eqref{eq:rescaled_equation} involves diverging products of $u_{\epsilon}$ as $\epsilon \rightarrow 0$, we recall from \cite{Hai14a}, \cite{HQ15} and \cite{HaiXu} that the strategy is to lift for each $\epsilon > 0$ the equation \eqref{eq:rescaled_equation} to an abstract space of regularity structures associated with a (random) model, solve the lifted equation as an abstract fixed point problem, and then project down to the ``real" solution with the reconstruction operator. An essential ingredient in this procedure is to bound arbitrarily high moments of the renormalised random models to prove their convergence. 

If the random models are all built from Gaussian noise, then by equivalence of moments, one only needs to check the bounds for the second moments, and those for all higher moments will follow immediately. However, for non-Gaussian noise, one has to bound the $p$-th moments of the random models by hand for all $p$, which involves complicated expressions of generalised convolutions of singular kernels. 

In \cite{HQ15}, the authors developed a graphic machinery for generalised convolutions, and reduced the estimates of these complicated convolutions to the verifications of certain conditions on labelled graphs. The $p$-th moment of a random object is the sum of certain graphs, each obtained by glueing together the nodes of $p$ identical trees in a certain way, and each of the identical trees represents that random object. When $p$ is large, it is in general very hard to keep track of all the conditions that need to be checked. In \cite{KPZCLT} and \cite{AjayHao}, the verification procedure was further reduced to checking the conditions for every sub-tree of a single tree (instead of a large graph consisting of $p$ trees) that represents the random object (see Assumption \ref{ass:ele-graph} and Theorem \ref{th:moment-bound} below). This significantly simplifies the verification; see for example \cite{KPZCLT} in the case of KPZ equation. 

However, in the current problem, Eq. \eqref{eq:rescaled_equation} will involve arbitrarily large powers of $u_{\epsilon}$ in contrast to \cite{KPZCLT} where only the second power appears. Thus, even for a single tree, it will be very difficult to check the conditions in Assumption \ref{ass:ele-graph} for  \textit{every} sub-tree of it. In addition, trees built from the non-linearity with high powers of $u_{\epsilon}$ will in general fail the integrability conditions (1) -- (3) in Assumption \ref{ass:ele-graph}. 

To overcome the second difficulty, we employ the positive ``homogeneities" represented by the multiplication of $\epsilon$, and incorporate them into the ``non-integrable" trees in a certain way to create ``integrable" trees (see \eqref{eq:macro_equation} and the beginning of Section \ref{sec:second_order_bounds}). On the other hand, to systematically and efficiently checking the conditions for these new ``integrable" trees, we make a key observation in Lemma \ref{le:simple-verification} below, which shows that one only needs to check the conditions for very few sub-trees, and the verification for all other sub-trees will automatically follow from that. We expect these treatments and simplifications will apply to other situations, for example proving universality of the KPZ equation for polynomial microscopic growth models with non-Gaussian noise. 

Before we state our main theorem, we first give precise assumptions on the noise $\zeta$ as well as the potential $V_{\theta}$. 

\begin{assumption} \label{ass:noise}
	We assume the random field $\zeta$ defined on $\R \times (\epsilon^{-1} \T)^{3}$ satisfies the following properties: 
	\begin{enumerate}
		\item $\zeta$ is symmetric (in the sense that $\zeta\stackrel{\text{law}}{=}-\zeta$), stationary, continuous, and $\E |\zeta (z)|^{p} < +\infty$ for all $p > 1$. 
		
		\item $\zeta$ is normalised so that $\varrho (z) = \E \zeta(0) \zeta (z)$ integrates to $1$. 
		
		\item For any two compact sets $K_{1}, K_{2} \subset \R \times (\epsilon^{-1} \T)^{3}$ with distance at least $1$ away from each other, the two sigma fields $\fF_{K_{1}}$ and $\fF_{K_{2}}$ are independent,
		where   $\fF_{K_{i}}$ is the $\sigma$-algebra generated by the point evaluations $\{\zeta(z)\,:\, z\in K_i\}$.
	\end{enumerate}
\end{assumption}

Typical examples of such random fields are smeared-out versions of Poisson point processes with uniform space-time intensity. We refer to \cite{KPZCLT} for more details of these and related examples. Given such a random field $\zeta$, we let $\Psi$ to be the stationary solution to the equation
\begin{equation} \label{eq:linearised}
\partial_{t} \Psi = \Delta \Psi + \zeta. 
\end{equation}
Since we are in dimension three, such a stationary solution exists as the square of the heat kernel is integrable at large scales. Let $\mu$ denote the distribution of $\Psi$ evaluated at a space-time point, and define the averaged potential $\scal{V_{\theta}}$ to be
\begin{equation} \label{eq:average_potential}
\scal{V_{\theta}}(x) := \int_{\R} V_{\theta}(x+y) \mu (dy). 
\end{equation}
Since $\zeta$ is assumed to be symmetric, 
$\mu$ is also symmetric, so $\scal{V_\theta}$ is the convolution
of $V_\theta$ and $\mu$.
Our main assumption on $V$ is the following. 

\begin{assumption} \label{ass:potential}
	Let $\mu$ denote the distribution of the stationary solution to \eqref{eq:linearised}, where $\zeta$ is a random field satisfying Assumption \ref{ass:noise}. Let $V: (\theta,x) \mapsto V_{\theta}(x)$ be a symmetric polynomial in $x$ with coefficients depending smoothly on $\theta$. 
	We assume the averaged potential $\scal{V_{\theta}}$ obtained from $V_{\theta}$ via \eqref{eq:average_potential} has a pitchfork bifurcation near the origin in the sense that
	\begin{equation} \label{eq:pitchfork_v}
	\frac{\partial^{4} \scal{V}}{\partial x^{4}} (0,0) > 0, \qquad \frac{\partial^{2} \scal{V}}{\partial x^{2}} (0,0) = 0, \qquad \frac{\partial^{3} \scal{V}}{\partial \theta \phantom{1} \partial x^{2}} (0,0) < 0. 
	\end{equation}
\end{assumption}

\begin{rmk}
	Note that it is always the derivative $V_{\theta}'$ rather than $V_{\theta}$ that appears in the dynamics, so the constant term of $V$ will not have any role in this context. Suppose
	\begin{equation} \label{eq:average_expression}
		\scal{V_{\theta}}'(x) = \sum_{j=0}^{m} \ha_{j}(\theta) x^{2j+1}, 
	\end{equation}
	then the pitchfork bifurcation assumption \eqref{eq:pitchfork_v} reads
	\begin{equation} \label{eq:pitchfork}
		\ha_{1} > 0, \qquad \ha_{0} = 0, \qquad \ha_{0}' < 0, 
	\end{equation}
	where we have used the notation $\ha_{j} = \ha_{j}(0)$ and $\ha_{j}' = \ha_{j}'(0)$ for simplicity. From now on, we will assume the averaged potential $\scal{V}$ has the expression \eqref{eq:average_expression}, and satisfies the pitchfork bifurcation assumption \eqref{eq:pitchfork}. 
\end{rmk}

Our main theorem can be loosely stated as follows.

\begin{thm} \label{th:main_loose}
Let $\zeta$ and $V$ satisfy Assumptions \ref{ass:noise} and \ref{ass:potential}. Let $u$ (depending on both $\epsilon$ and $\theta$) be the solution to \eqref{eq:micro_model}. Then, there exists $a < 0$ such that for $\theta = a \epsilon |\log \epsilon| + \oO(\epsilon)$, the rescaled process $u_{\epsilon}$ in \eqref{eq:rescaled_solution} converges to the $\Phi^4_3 (\ha_{1})$ family solutions. 
\end{thm}

\begin{rmk}
The symmetry assumption on both the potential and the noise is essential to get $\Phi^4_3$ limit. In \cite{HaiXu}, it was shown that with Gaussian (symmetric) noise but asymmetric potential, the system would in general rescale to either OU process or $\Phi^3_3$, depending on the value of $\theta$. We expect the same phenomena to happen with asymmetric noise even if the potential is symmetric. 
\end{rmk}

Before we proceed, we first briefly explain why the pitchfork bifurcation assumption naturally appears on $\scal{V}$ but not $V$. To see this, we note that if $u$ solves the microscopic equation \eqref{eq:micro_model}, then the rescaled process $u_{\epsilon}$ defined in \eqref{eq:rescaled_solution} solves the equation
\begin{equation} \label{eq:macro_equation}
\partial_{t} u_{\epsilon} = \Delta u_{\epsilon} - \sum_{j=0}^{m} a_{j}(\theta) \epsilon^{j-1} u_{\epsilon}^{2j+1} + \zeta_{\epsilon}, 
\end{equation}
where $a_{j}(\theta)$'s are the coefficients of $V_{\theta}'$, and
\begin{equation} \label{eq:rescaled_field}
\zeta_{\epsilon} (t,x) = \epsilon^{-\frac{5}{2}} \zeta(t/\epsilon^{2}, x/\epsilon)
\end{equation}
is the rescaled noise. Heuristically, as $u_{\epsilon}^{2} \sim \epsilon^{-1}$ as $\epsilon \rightarrow 0$, so each of the terms $\epsilon^{j-1} u_{\epsilon}^{2j+1}$ ($j \geq 1$) has a ``diverging part" of the order $\epsilon^{-1} u_{\epsilon}$, and the combined effect of these divergences is then a constant multiple of that order. It turns out that the value of this constant is nothing but
\begin{align*}
\int_{\R} V_{0}'(y) \mu(dy) =: \ha_{0}, 
\end{align*}
so $\ha_{0} = 0$ in \eqref{eq:pitchfork} is a necessary condition for the right hand side of \eqref{eq:macro_equation} to converge. In fact, as we will see later, with $\ha_{0} = 0$, the Wick renormalisation (for non-Gaussian noise $\zeta_{\epsilon}$) is automatically taken account of, and thus the coefficients of these divergent terms with order $\epsilon^{-1}$ will precisely cancel out each other. 

Also, as in the standard $\Phi^4_3$ case, there will be further logarithmic divergence beyond the Wick ordering. But this divergence can be renormalised by a suitable choice of the small parameter $\theta$ thanks to $\ha_{0}' \neq 0$. Finally, $\ha_{1} > 0$ is needed for the final limit to be the $\Phi^4_3$ family. 

The article is organised as follows. In Section \ref{sec:Wick}, we briefly recall the facts of non-Gaussian random variables and their Wick powers. Section \ref{sec:setup} is devoted to the construction of the regularity structures, the abstract fixed point equation as well as the renormalisation group. Most of the set-up in these sections could be found in \cite{Hai14a}, \cite{KPZCLT}, and \cite{HaiXu}, so we keep the presentation to a minimum and refer to those articles for more details. 

In Section \ref{sec:bounds}, we prove the bounds for the renormalised models constructed in Section \ref{sec:renormalisation} with proper renormalisation constants. These bounds are necessary in proving the convergence of renormalised models. Finally, in Section \ref{sec:limit}, we collect the results in all previous sections to prove the convergence of our models to the desired $\Phi^4_3$ limit at large scales.

\subsection*{Acknowledgements}

{
We would like to thank 
Ajay Chandra for discussions on general moment bounds. WX acknowledges the Mathematical Sciences Research Institute in Berkeley, California for its warm hospitality, and the National Science Foundation for supporting his stay here. 
}

\section{Non-Gaussian Wick powers and averaged potential} \label{sec:Wick}

In this section, we review some basic properties of Wick polynomials for non-Gaussian random variables. We keep the presentation to the minimum here and only give definitions and statements that will be used later. More details can be found in Section 3 of \cite{KPZCLT}. 

\subsection{Joint cumulants} \label{sec:cumulants}

We start with the definition of joint cumulants. Given a collection of random variables $\xX = \{X_\alpha\}_{\alpha \in A}$ for some index set $A$ and a subset $B \subset A$, we write $X_B \subset \xX$ and $X^B$ as short-hands for
\begin{equation}
	X_B = \{X_\alpha\,:\, \alpha \in B\} \;,\qquad X^B = \prod_{\alpha\in B}X_\alpha\;.
\end{equation}
Given a finite set $B$, we furthermore write $\CP(B)$ for the collection of all partitions of $B$, i.e.\ all sets $\pi \subset \powerset(B)$ (the power set of $B$) such that $\bigcup \pi = B$ and such that any two distinct elements of $\pi$ are disjoint. 

\begin{defn} \label{def:cumu}
	Given a collection $\xX$ of random variables as above and any finite set 
	$B \subset A$, we define the cumulant 
	$\Cum(X_B)$ inductively over $|B|$ by
	\begin{equation} \label{e:mome2cumu}
		\E \big( X^B\big)
		= \sum_{\pi \in \CP(B)} \prod_{\bar B\in\pi} \Cum \big(X_{\bar B}\big)\;.
	\end{equation}
\end{defn}

\begin{rmk}
It is straightforward to check that the above definition does determine the cumulants uniquely. 
\end{rmk}

From now on, we will use the notation $\fC_n$ for the $n$th joint cumulant function of the field $\zeta$: 
\begin{equation} \label{e:def-kappa}
	\fC_n (z_1,\ldots,z_n) = \Cum\bigl(\{\zeta(z_1),\cdots, \zeta(z_n)\}\bigr) \;.
\end{equation}
We similarly write $\fC_{n}^{(\epsilon)}$ for the $n$-th cumulant but with $\zeta(z_{j})$'s replaced by $\zeta_{\epsilon}(z_{j})$'s. Note that $\fC_{2n+1} = 0$ since $\zeta$ is assumed to be symmetric. Also, $\fC_2$ is its covariance function.
An important property is that  there exists $C > 0$ such that
\begin{equation} \label{eq:cumu-close}
	\fC_p^{(\epsilon)}(z_{1}, \cdots, z_{p})  = 0
\end{equation}
whenever $|z_{i} - z_{j}| > C \epsilon$ for some $i,j$. 

%
%
%

\subsection{Wick polynomials} \label{sec:wick_polynomial}

The notion of Wick products for random variables (not necessarily Gaussian) will play an essential role later. We give a definition below. 

\begin{defn} \label{def:Wick-product}
	Given a collection $\CX = \{X_\alpha\}_{\alpha \in \CA}$ of random variables as before, the Wick product 
	$\Wick{X_A}$ for $A \subset \CA$ is by setting $ \Wick{X_{\emptyset}} = 1$ 
	and postulating that
	\begin{equation} \label{e:defWick}
		X^A  = 
		\sum_{B \subset A}  \Wick{X_B}  \sum_{\pi \in \CP(A \setminus B)} 
		\prod_{\bar B \in \pi}
		\Cum \big(X_{\bar B}\big)   
	\end{equation}
	recursively. 
\end{defn}

\begin{rmk}
Again, \eqref{e:defWick} is sufficient to define $\Wick{X_A}$ by induction over the size of $A$. 
By the definition we can easily see that as soon as $A \neq \emptyset$, one always has $\E \Wick{X_A}=0$.
Note also that  taking expectations on both sides of \eqref{e:defWick} yields exactly the identity \eqref{e:mome2cumu}.
\end{rmk}

Note that by \eqref{e:mome2cumu}, we can alternatively write 
$X^A  =\sum_{B \subset A}  \Wick{X_B} \E(X^{A\setminus B})$,
and there is also a formula to express the Wick product in terms of the usual products;
however, \eqref{e:defWick} is actually the identity frequently being used in the paper: 
in fact we will frequently rewrite a product of generally non-Gaussian noises as a sum of terms,
each term containing a Wick product as RHS of \eqref{e:defWick} - called the 
{\it non-Gaussian homogeneous chaos of order $|B|$}.

It is also well-known that there is an alternative characterisation of Wick products via generating functions by
\begin{equation} \label{e:gen-Wick}
	\Wick{X_{1} \cdots X_{n}} = {\partial^n \over \partial t_1\cdots\partial t_n} 
	{\exp \big( \sum_{i=1}^n t_i X_i\big) \over \E \exp \big( \sum_{i=1}^n t_i X_i\big)}
	\Big|_{t_1=\ldots=t_n=0} \;.
\end{equation}
In this case $X_{1} = \cdots =X_{n}$ all with distribution $\mu$, then \eqref{e:gen-Wick} can be reduced to
\begin{equation} \label{eq:identical_wick}
\Wick{X^{n}} = \frac{\partial^{n}}{\partial t^{n}} \bigg( \frac{e^{tX}}{ \E_{\mu} e^{t Y}} \bigg) \bigg|_{t=0}
\end{equation}
where $Y$ is distributed according to $\mu$.
Note that there is no $n$ in the exponential generating function above, as the derivative is taken with respect to the same $t$ rather than $t_{1}, \cdots, t_{n}$ separately. The form of \eqref{eq:identical_wick} also suggests that we can actually define the $n$-th Wick power with respect to a measure $\mu$ as
\begin{equation} \label{eq:gen_Wick}
	W_{n,\mu}(x) = {\partial^n \over \partial t^n} 
	{\exp \big(  t x \big) \over \E_{\mu} \exp \big(  t Y \big)}
	\Big|_{t=0} \;.
\end{equation}
Now, $W_{n,\mu}(\cdot)$ is a polynomial in the variable $x \in \R$, rather than the random variable $X$. 
One can immediately check using the definition \eqref{eq:gen_Wick}
that $\frac{d}{dx}W_{n,\mu}(x)=nW_{n-1,\mu}$,
and this means that the Wick powers form an Appell sequence.
As a comparison to the Wick powers, we recall that the usual monomials can be generated by
\begin{equation} \label{eq:gen_usual}
	x^n= {\partial^n \over \partial t^n} 
	e^{ t x} 
	\Big|_{t=0} \;.
\end{equation}
We then have the following important lemma.  

\begin{lem} \label{lem:ave-potential}
	Let $\mu$ be a probability measure on $\R$ with finite moments of all order, and $V$ be a polynomial. Let
	\begin{align*}
		\scal{V}(x) = \int_{\R} V(x+y) \mu (dy) = \E_{\mu} V(x+Y)
	\end{align*}
	be the average of $V$ against the measure $\mu$. Then, $\scal{V}(x) = x^{n}$ if and only if $V(x) = W_{n,\mu}(x)$. 
\end{lem}
\begin{proof}
	We first prove the ``if" part.  Let $V(x) = W_{n,\mu}(x)$, then by definitions of $W_{n,\mu}$ and $\scal{\cdot}$, we have
	\begin{align*}
		\scal{V}(x) = \E_{\mu} \frac{\partial^{n}}{\partial t^{n}} \Big(\frac{e^{tx} e^{tY}}{\E_{\mu} e^{tZ}} \Big) \Big|_{t=0} = \frac{\partial^{n}}{\partial t^{n}} \Big(\frac{e^{tx} \E_{\mu} e^{tY}}{\E_{\mu} e^{tZ}} \Big) \Big|_{t=0} = \frac{\partial^{n}}{\partial t^{n}} e^{tx} \Big|_{t=0} = x^{n}, 
	\end{align*}
	where for the second equality we have exchanged the differentiation with the expectation with respect to $Y$. For the ``only if" part, suppose $U$ is another polynomial with $\scal{U} = x^{n}$, then $U-V$ is a polynomial such that
	\begin{align*}
	\int_{\R} (U-V)(x+y) \mu(dy) = 0
	\end{align*}
	for all $x \in \R$. It then follows easily that all coefficients of $U-V$ must be $0$, and we have $U=V=W_{n,\mu}$. 
\end{proof}

\section{Regularity structures and abstract equation} \label{sec:setup}

In this section, we build the appropriate regularity structures in order to solve the equation
\begin{equation} \label{eq:real_equation}
\partial_{t} u_{\epsilon} = \Delta u_{\epsilon} - \epsilon^{-\frac{3}{2}} V_{\theta}'(\epsilon^{\frac{1}{2}}u_{\epsilon}) + \zeta_{\epsilon}
\end{equation}
in an abstract space. The set-up follows essentially the same way as in \cite{HaiXu} and \cite{phi4review}. 

\subsection{Regularity structures and admissible models} \label{sec:rs}

Recall that {\it a regularity structure} consists of a pair $(\tT, \gG)$, where $\tT = \bigoplus_{\alpha \in A} \tT_{\alpha}$ is a vector space graded by a (bounded below, locally finite) set $A$ of homogeneities, and $\gG$ is a group of linear transformations on $\tT$ such that for every $\Gamma \in \gG$ and $\tau \in \tT_{\alpha}$, we have $\Gamma \tau - \tau \in \tT_{< \alpha}$. 

Similar as \cite{HaiXu}, the basis elements in $\tT$ are built from the symbols $\1, \Xi$, $\{X_{i}\}_{i=0}^{3}$ and the abstract integration operators $\iI$ and $\eE^{k}$ (unlike \cite{HaiXu}, we only need integer powers of $\eE$ here). In order to reflect the parabolic scaling of the equation, for any multi-index of non-negative integers $k = (k_{0}, \cdots, k_{3})$, we define an abstract monomial of degree $|k|$ as
\begin{equation} \label{eq:parabolic_degree}
X^{k} = X_{0}^{k_{0}} \cdots X_{3}^{k_{3}}, \qquad |k| = 2 k_{0} + \sum_{i=1}^{3} k_{i}, 
\end{equation}
and the parabolic distance between $z,z' \in \R^{4}$ by
\begin{equation} \label{eq:parabolic_distance}
|z-z'| = |z_{0} - z_{0}'|^{\frac{1}{2}} + \sum_{i=1}^{3} |z_{i} - z_{i}'|. 
\end{equation}
In view of the structure of \eqref{eq:real_equation}, we define two sets $\uU$ and $\vV$ such that $X^{k} \in \uU$, $\Xi \in \vV$, and such that for every $k = 1, \cdots, m$, we decree
\begin{equation} \label{eq:graded_vector}
\begin{split}
\tau_{1}, \cdots \tau_{2k+1} \in \uU \qquad &\Rightarrow \qquad \eE^{k-1}(\tau_{1} \cdots \tau_{2k+1}) \in \vV, \\
\tau \in \vV \qquad &\Rightarrow \qquad \iI(\tau) \in \uU. 
\end{split}
\end{equation}
The idea is that $\uU$ consists of the formal symbols in the expansion of the abstract solution of the equation, and $\vV$ consists of symbols that will appear in the expansion of the right hand side of \eqref{eq:real_equation}. We then let $\wW = \uU \cup \vV$ to be the set of basis elements in $\tT$. As for the homogeneities, we set
\begin{align*}
|\Xi| = -\frac{5}{2} - \kappa, \qquad |X^{\ell}| = |\ell|, 
\end{align*}
and extend it to all of $\wW$ decreeing
\begin{align*}
|\tau \bar{\tau}| = |\tau| + |\bar{\tau}|, \quad |\iI(\tau)| = |\tau| + 2, \quad |\eE^{k}(\tau)| = k + |\tau|. 
\end{align*}
Here, $\kappa$ is a small but fixed positive number, and $|\ell|$ denotes the parabolic degree of the multi-index defined in \eqref{eq:parabolic_degree}. According to \eqref{eq:graded_vector}, we can keep adding new basis elements into $\uU$ and $\vV$, but for the purpose of this article, it suffices to restrict our regularity structures to elements with homogeneities less than $\frac{3}{2}$, i.e., $\tT = \bigoplus_{\alpha < \frac{3}{2}} \tT_{\alpha}$.

The basis vectors in $\wW$ generated this way with negative homogeneities other than $\Xi$ are of the form (with shorthand $\Psi=\iI(\Xi)$): 
\begin{equs} \label{e:simbols}
\tau  = \eE^{k-1} \Psi^{2k+1-n} \;, \qquad \qquad \qquad
	 &|\tau| =\frac12(n-3) -(2k+1-n)\,\kappa \;;\\
\tau = \eE^{\lfl (k-1)/2 \rfl} \big( \Psi^{k} \iI(\eE^{\lfl \ell/2  \rfl - 1} \Psi^{\ell}) \big) \;, \quad
	&|\tau| = \big\lfloor \frac{k-1}{2} \big\rfloor 
		+ \big\lfloor \frac{\ell}{2}  \big\rfloor 
		- \big(k+\ell \big)\big(\frac12+\kappa\big)+1
\end{equs}
provided $\kappa$ is sufficiently small, where in the first case $n \in \{0,1,2,3\}$, and in the second one the situation when $k$ odd and $\ell$ even is excluded. 

As for the transition group $\gG$, we define the extended regularity structure $\tT_{\text{ex}}$ 
and a linear map $\Delta: \tT_{\text{ex}} \rightarrow \tT_{\text{ex}} \otimes \tT_{+}$ in the same way as in \cite[Section 3.1]{HQ15}. For any linear functional $g$ on $\tT_{+}$, one can obtain a linear map $\Gamma_{g}: \tT_{\text{ex}} \rightarrow \tT_{\text{ex}}$ by setting $\Gamma_{g} \tau = (\id \otimes g) \Delta \tau$. The transition group $\gG$ is defined by the collection of all $\Gamma_{g}$'s for linear functional $g$ on $\tT$ with the further property that $g(\tau \bar{\tau}) = g(\tau) g(\bar{\tau})$. The restriction of $\gG$ on $\tT$ gives the regularity structure $(\tT, \gG)$. 

\begin{rmk}
	The graded vector space of extended regularity structure $\tT_{\text{ex}}$ is the linear span of symbols in
	\begin{align*}
	\wW \cup \{ \tau_{1}, \cdots, \tau_{2m+1}: \tau_{j} \in \uU \}. 
	\end{align*}
	The main advantage of introducing $\tT_{\text{ex}}$ is that $\eE^{k}$ can be viewed as a linear map defined on $\tT_{\text{ex}}$. As a consequence, the definition of the renormalisation group in Section \ref{sec:renormalisation} will be significantly simplified. 
\end{rmk}

Given a regularity structure $(\tT, \gG)$, {\it a model for it} is a pair $(\Pi, F)$ of functions
\begin{align*}
\Pi: \R^{1+3} &\rightarrow \lL(\tT, D') \quad \quad F: \R^{1+3} \rightarrow \gG \\
z &\mapsto \Pi_{z} \qquad \qquad \qquad \quad z \mapsto F_{z}
\end{align*}
satisfying the identity
\begin{align*}
\Pi_{z} F_{z}^{-1} = \Pi_{\bar{z}} F_{\bar{z}}^{-1}, \qquad \forall z,\bar{z} \in \R^{1+3}
\end{align*}
as well as the bounds
\begin{equation} \label{eq:model_analytic}
|(\Pi_{z} \tau)(\varphi_{z}^{\lambda})| \lesssim \lambda^{|\tau|}, \qquad |F_{z}^{-1} \circ F_{\bar{z}}|_{\sigma} \lesssim |z-\bar{z}|^{|\tau| - |\sigma|}
\end{equation}
uniformly over all test functions $\varphi \in \bB$, all space-time points $z,\bar{z}$ in a compact domain $\fK$, and all $|\tau| \in \wW$. Here, the set $\bB$ denotes the space of test functions
\begin{align*}
\bB = \big\{  \varphi: \| \varphi \|_{\cC^{2}} \leq 1, \text{supp}(\varphi) \subset B(0,1) \big\},
\end{align*}
where $B(0,1)$ is the unit ball centred at origin, and $\varphi_{t',x'}^{\lambda}(t,x) = \lambda^{-5} \varphi\big((t-t')/\lambda^{2}, (x-x')/\lambda\big)$. We will also use the notation $f_{z}$ for the multiplicative element in $\tT_{+}^{*}$ (the adjoint of $\tT_{+}$) such that $F_{z} = \Gamma_{f_{z}}$. 

The norm of a model $\fM = (\Pi, f)$ is defined to be the smallest proportionality constant such that both bounds in \eqref{eq:model_analytic} hold, which we will denote by $|\!|\!| \fM |\!|\!|$. The norm in general depends on the compact domain $\fK$, but we omit it from the notation just for simplicity. Since in most situations the transition group $F$ is completely determined by $\Pi$, we will also sometimes write $|\!|\!| \Pi |\!|\!|$ instead of $|\!|\!| \fM |\!|\!|$. 

In addition to these minimal requirements, we also impose the \text{admissibility} on our models. Let $K: \R^{1+3} \setminus \{ 0 \} \rightarrow \R$ be a function that coincides the heat kernel on $\{|z| \leq 1\}$, smooth everywhere except the origin, and annihilates all polynomials on $\R^{1+3}$ up to parabolic degree $2$. The existence of such a kernel can be easily checked (see for example \cite[Section 5]{Hai14a}). 

With this truncated heat kernel $K$, we then let $\sM$ to be the set of \text{admissible models} $(\Pi, f)$ such that one has
\begin{equation} \label{eq:model_polynomial}
(\Pi_{z} X^{k})(\bar{z}) = (\bar{z}-z)^{k}, \qquad f_{z}(X^{k}) = (-z)^{k}, \qquad \forall k
\end{equation}
as well as 
\begin{equation} \label{eq:model_integration}
\begin{split}
f_{z} (\sJ_{\ell} \tau) &= - \int D^{\ell} K(z-\bar{z}) (\Pi_{z} \tau) (d \bar{z}), \qquad 0 \leq |\ell| < |\tau| + 2 \\
\big( \Pi_{z} \iI(\tau) \big) (\cdot)  &= (K * \Pi_{z} \tau)(\cdot) + \sum_{0 \leq |\ell| < |\tau|+2} \frac{(\cdot - z)^{\ell}}{\ell !} \cdot f_{z} (\sJ_{\ell} \tau), 
\end{split}
\end{equation}
and we set $f_{z} (\sJ_{\ell} \tau) = 0$ if $|\ell| \geq |\tau| + 2$. Note that the above notion of admissible models do not impose any restrictions on the operator $\eE^{k}$. 

Given a smooth function $\zeta$ and $\epsilon \geq 0$, we now build {\it the canonical model} $\sL_{\epsilon}(\zeta) = (\Pi^{\epsilon},f^{\epsilon})$ as follows. Define the action of $\sL_{\epsilon}(\zeta)$ on $X^{k}$ according to \eqref{eq:model_polynomial}, and set set
\begin{align*}
(\Pi_{z} \Xi) (\bar{z}) = \zeta (\bar{z}). 
\end{align*}
We then extend the definition to all of $\wW$ by imposing the admissibility restriction \eqref{eq:model_integration}, as well as setting recursively
\begin{equation} \label{eq:canonical_multiplication}
(\Pi_{z}^{\epsilon} \tau \bar{\tau})(\bar{z}) = (\Pi_{z}^{\epsilon} \tau)(\bar{z}) \cdot (\Pi_{z}^{\epsilon} \bar{\tau})(\bar{z})
\end{equation}
and
\begin{equation} \label{eq:model_epsilon}
\begin{split}
f_{z}^{\epsilon}(\sE^{k}_{\ell} \tau) &= -\epsilon^{k} (D^{\ell} \Pi_{z} \tau)(z), \\
(\Pi_{z}^{\epsilon} \eE^{k} \tau)(\bar{z}) &= \epsilon^{k} (\Pi_{z}^{\epsilon} \tau)(\bar{z}) + \sum_{0 \leq |\ell| < k + |\tau|} \frac{(\bar{z}-z)^{\ell}}{\ell !} \cdot f_{z}^{\epsilon} (\sE^{k}_{\ell} \tau), 
\end{split}
\end{equation}
where we adopt that $\sE^{k}_{\ell} \tau$ exists only when $|\ell| < k + |\tau|$. Note that the right hand sides of \eqref{eq:canonical_multiplication} and \eqref{eq:model_epsilon} make sense only if $\Pi_{z}^{\epsilon} \tau$ is a smooth function for every $\tau$, and this is indeed the case if the input $\zeta$ is smooth. 

\begin{rmk}
	The action of models on the basis elements $\sJ_{\ell} \tau$ and $\sE^{k}_{\ell} \tau$ are essentially the same as the $\ell$-th derivative of the realisations for $\iI(\tau)$ and $\eE^{k}(\tau)$. But we use a different notation for these elements since they will play a different role than their corresponding elements in $\tT$. 
\end{rmk}

\subsection{$\epsilon$-dependent spaces and the abstract fixed point equation}

Since the solution to standard $\Phi^4_3$ equation possesses regularity below $-\frac{1}{2}$, in a fixed point argument, in order to continue local solutions, we need to be able to treat initial data of the same regularity. However, if the initial data belongs to $\cC^{\eta}$ with $\eta < -\frac{1}{2}$, then $u_{\epsilon}(t, \cdot) \sim t^{-\frac{\eta}{2}}$ for small time $t$, and for any fixed $\epsilon$, the term $\epsilon^{k-1} u_{\epsilon}^{2k+1}$ becomes non-integrable as soon as $k \geq 2$. 

On the other hand, the case we are mostly interested in is the limiting equation as $\epsilon \rightarrow 0$. The idea implemented in \cite{HQ15} to achieve this goal is to employ cancellations between positive power of $\epsilon$ and singularities in $t$ in the process as $\epsilon \rightarrow 0$ while retaining uniform (in $\epsilon$) bounds. This leads to the definition of $\epsilon$-dependent models and spaces. The definitions below mainly follow \cite{HQ15} and \cite{HaiXu}. 

For each $\epsilon \geq 0$, let $\sM_{\epsilon}$ be the collections of models in $\sM$ with norm
\begin{equation} \label{eq:e_norm}
|\!|\!| \Pi |\!|\!|_{\epsilon}:= |\!|\!| \Pi |\!|\!| + \| \Pi \|_{\epsilon} + \sup_{t \in [0,1]} \|\rR \iI(\Xi)(t, \cdot) \|_{\cC^{\eta}}, 
\end{equation}
where $|\!|\!| \cdot |\!|\!|$ is the standard norm on modelled distributions, $\rR$ is the reconstruction operator associated to the underlying model (as defined in \cite[Section 3]{Hai14a}), and $\| \cdot \|_{\epsilon}$ is defined by
\begin{equation} \label{eq:e_partial_norm}
\begin{split}
\| \Pi \|_{\epsilon} &:= \sup_{\tau} \sup_{z} \sup_{k,\ell} \epsilon^{|\ell|-k-|\tau|} |f_{z}(\sE^{k}_{\ell}(\tau))| \\
&\phantom{11}+ \sup_{|\tau| \in \uU} \sup_{z} \sup_{\psi} \lambda^{-\beta} \epsilon^{\beta - |\tau|} |(\Pi_{z}\tau)(\psi_{z}^{\lambda})|, 
\end{split}
\end{equation}
where $\beta = \frac{6}{5}$, and the supremum is taken over all $\psi \in \bB$ that further annihilates affine functions. We also let $\sM_{0}$ denote the set of admissible models such that $f_{z}(\sE^{k}_{\ell} \tau) = 0$ for all $\tau \in \wW$. Note that for any positive $\epsilon$ and $\bar{\epsilon}$, $\sM_{\epsilon}$ and $\sM_{\bar{\epsilon}}$ consists of exactly the same set of models, but with norms at different scales. We compare two models $\Pi^{\epsilon} \in \sM_{\epsilon}$ and $\Pi \in \sM_{0}$ by
\begin{equation} \label{eq:e_norm_compare}
|\!|\!| \Pi^{\epsilon}; \Pi |\!|\!|_{\epsilon;0} := |\!|\!| \Pi^{\epsilon}; \Pi |\!|\!| + \| \Pi^{\epsilon} \|_{\epsilon} + \sup_{t \in [0,1]} \| \rR^{\epsilon} \iI(\Xi)(t, \cdot) - \rR \iI(\Xi)(t, \cdot) \|_{\cC^{\eta}}. 
\end{equation}
We also define the $\epsilon$-weighted function space $\cC^{\gamma,\eta}_{\epsilon}$ and modelled distribution space $\dD^{\gamma,\eta}_{\epsilon}$ in exactly the same way as \cite[Section 3.1]{HaiXu}. One way to understand the space $\cC^{\gamma,\eta}_{\epsilon}$ is that functions in it behaves like $\cC^{\eta}$ on scales larger than $\epsilon$, but like $\cC^{\gamma}$ on scales smaller than $\epsilon$. It is also suggestive to think of $\dD^{\gamma,\eta}_{\epsilon}$ as the analogue for modelled distributions. 

Let $P$ denote the heat kernel, and $K$ be its truncation as defined in Section \ref{sec:rs}. Then by \cite[Section 4]{Hai14a}, there exists an operator $\pP$ on $\dD^{\gamma,\eta}_{\epsilon}$ such that
\begin{align*}
\rR (\pP f) = P * \rR f.
\end{align*}
 Finally, we define the operator $\heE^{k}$ (mapping $\dD^{\gamma,\eta}_{\epsilon}$ to $\dD^{\gamma',\eta'}_{\epsilon}$ for some other $\gamma', \eta'$) by
\begin{align*}
(\heE^{k} U)(z) = \eE^{k} U(z) - \sum_{\ell} \frac{X^{\ell}}{\ell !} f_{z}(\sE^{k}_{\ell} U(z)). 
\end{align*}
Since the only use of these spaces and the operator $\heE_{k}$ in this article is in the statement of Theorem \ref{th:abstract} below, we do not repeat their properties here, but refer to \cite{HQ15} and \cite{HaiXu} for more details. 

The following theorem is identical to \cite[Theorem 3.12]{HaiXu}. It gives the existence, uniqueness and convergence of abstract solutions to the lift of the equation \eqref{eq:real_equation} in the regularity structures. 

\begin{thm} \label{th:abstract}
	Let $m \geq 1$, $\gamma \in (1,\frac{6}{5})$, $\eta \in (-\frac{2m+2}{4m+3},-\frac{1}{2})$, and $\kappa > 0$ be sufficiently small. Let $\phi_{0} \in \cC^{\gamma,\eta}_{\epsilon}$, and consider the equation
	\begin{equation} \label{eq:abstract}
	\Phi = \pP \1_{+} \bigg( \Xi - \sum_{j=1}^{m} \lambda_{j} \heE^{j-1} \Phi^{2j+1} - \lambda_{0} \Phi \bigg) + \widehat{P} \phi_{0}. 
	\end{equation}
	Then, for every sufficiently small $\epsilon$ and every model in $\sM_{\epsilon}$, there exists $T > 0$ such that Equation \eqref{eq:abstract} has a unique solution in $\dD^{\gamma,\eta}_{\epsilon}$ up to time $T$. Moreover, $T$ can be chosen uniformly over any fixed bounded set of initial data in $\cC^{\gamma,\eta}_{\epsilon}$, any bounded set of models in $\sM_{\epsilon}$, and all sufficiently small $\epsilon$. 
	
	Let $\phi_{0}^{(\epsilon)}$ be a sequence of elements in $\cC^{\gamma,\eta}_{\epsilon}$ such that $\| \phi_{0}^{(\epsilon)}; \phi_{0} \|_{\gamma,\eta;\epsilon} \rightarrow 0$ for some $\phi_{0} \in \cC^{\eta}$, $\Pi^{\epsilon} \in \sM_{\epsilon}, \Pi \in \sM_{0}$ be models such that $|\!|\!| \Pi^{\epsilon}; \Pi |\!|\!|_{\epsilon;0} \rightarrow 0$, and let $\lambda_{j}^{(\epsilon)} \rightarrow \lambda_{j}$ for each $j$. If $\Phi \in \dD^{\gamma,\eta}_{0}$ solves the fixed point problem \eqref{eq:abstract} up to time $T$ with model $\Pi$, initial data $\phi_{0}$ and coefficients $\lambda_{j}$, then for all small enough $\epsilon$, there is a unique solution $\Phi^{(\epsilon)} \in \dD^{\gamma,\eta}_{\epsilon}$ to \eqref{eq:abstract} up to the same time $T$ with initial data $\phi_{0}^{(\epsilon)}$, model $\Pi^{\epsilon}$ and coefficients $\lambda_{j}^{(\epsilon)}$'s, and we have 
	\begin{align*}
	\lim_{\epsilon \rightarrow 0}  |\!|\!| \Phi^{(\epsilon)}; \Phi |\!|\!|_{\gamma,\eta;\epsilon} \rightarrow 0. 
	\end{align*}
\end{thm}

\subsection{The renormalisation group} \label{sec:renormalisation}

If $\Phi$ solves the abstract equation \eqref{eq:abstract} with the canonical model $\sL_{\epsilon}(\zeta_{\epsilon})$, then $u_{\epsilon} = \rR^{\epsilon} \Phi$ solves the PDE
\begin{align*}
\partial_{t} u_{\epsilon} = \Delta u_{\epsilon} - \sum_{j=1}^{m} \lambda_{j} \epsilon^{j-1} u_{\epsilon}^{2j+1} - \lambda_{0} u_{\epsilon} + \zeta_{\epsilon}. 
\end{align*}
But as $\epsilon \rightarrow 0$, the models $\sL_{\epsilon}(\zeta_{\epsilon})$ simply do not converge, and hence there is no clear meaning to define the limit of $u_{\epsilon}$. This means that we need to perform suitable \text{renormalisations} on the equation. 

The purpose of this section is to build a family $\fR$ of linear maps on $\tT_{\text{ex}}$ so that when restricted to basis elements in $\tT$, their action on admissible models yields the ``renormalised models" (see \eqref{eq:renormalised_model} below) that will converge, and the reconstructed solution to the fixed point equation \eqref{eq:abstract} satisfies some modified PDE. 

For these purposes, the type of linear transformations $M \in \fR$ we consider will be the composition of two different maps $M_{0}$ and $\wM$, which acts on the original model $(\Pi, f)$ in the following way: 
\begin{align*}
(\Pi, f) \stackrel{\wM}{\longmapsto} (\wPi, \wf) \stackrel{M_{0}}{\longmapsto} (\Pi^{M}, f^{M}). 
\end{align*}
Here, $\wM$ will act on the models as ``Wick renormalisation", induced by a probability measure on $\R$, and $M_{0}$ has the interpretation as ``mass renormalisation" in quantum field theory, and will be parametrised by finitely many parameters.

We start with the description of the Wick renormalisation map. For any probability measure $\mu$ on $\R$ with all moments finite, we define the linear map $\wM_{\mu}$ by setting
\begin{align*}
\wM_{\mu} \Xi = \Xi, \qquad \wM_{\mu} X^{k} = X^{k}
\end{align*}
and
\begin{align*}
\wM_{\mu} \Psi^{n} = W_{n,\mu}(\Psi), \qquad \Psi = \iI(\Xi)
\end{align*}
mapping $\Psi^{n}$ into the Wick polynomial induced by $\mu$. We furthermore require that $\wM_{\mu}$ commutes with the abstract integration maps $\iI$ and $\eE^{k}$ as well as the multiplication by $X^{k}$, and extend it to the whole of $\tT_{\text{ex}}$ by 
\begin{equation} \label{eq:leibniz}
\wM_{\mu} (\tau \iI (\bar{\tau})) = \wM_\mu (\tau) \iI (\wM_\mu \bar{\tau}). 
\end{equation}
For the map $\wM_\mu$ defined above, it follows from \cite{Hai14a, HQ15, HaiXu} that there is a unique pair of linear maps
\begin{align*}
\wDelta: \tT_{\text{ex}} \rightarrow \tT_{\text{ex}} \otimes \tT_{+}, \qquad \widehat{M}^{\text{Wick}}: \tT_{+} \rightarrow \tT_{+}
\end{align*}
such that
\begin{equation} \label{eq:wick_association}
\begin{split}
\widehat{M}^{\text{Wick}} \sJ_{\ell} &= \mM (\sJ_{\ell} \otimes \id) \wDelta, \\
\widehat{M}^{\text{Wick}} \sE^{k}_{\ell} &= \mM (\sE^{k}_{\ell} \otimes \id) \wDelta, \\
(\id \otimes \mM) (\Delta \otimes \id) \wDelta &= (\wM_\mu \otimes \widehat{M}^{\text{Wick}}) \Delta, \\
\widehat{M}^{\text{Wick}}(\tau_{1} \tau_{2}) &= (\widehat{M}^{\text{Wick}} \tau_{1}) (\widehat{M}^{\text{Wick}} \tau_{2}), \quad \widehat{M}^{\text{Wick}} X^{k} = X^{k},   
\end{split}
\end{equation}
where $\mM: \tT_{+} \rightarrow \tT_{+}$ denotes the multiplication in Hopf algebra. Given an admissible model $(\Pi, f)$, we define its \textit{Wick renormalised model} $(\wPi, \wf)$ by
\begin{equation} \label{eq:wick_model}
\wPi_{z}\tau = (\Pi_{z} \otimes f_{z}) \wDelta \tau, \qquad \wf_{z} (\sigma) = f_{z} (\widehat{M}^{\text{Wick}} \sigma). 
\end{equation}
The following proposition ensures that the new pair $(\wPi, \wf)$ is indeed an admissible model. 

\begin{prop} \label{pr:upper-triangular}
The unique map $\wDelta$ defined in \eqref{eq:wick_association} has the following upper-triangular property: for every $\tau \in \tT$, one has
\begin{equation} \label{eq:upper-triangular}
\wDelta \tau = \tau \otimes \1 + \sum \tau^{(1)} \otimes \tau^{(2)}, 
\end{equation}
where each term in the sum satisfies $|\tau^{(1)}| > |\tau|$. As a consequence, the new pair $(\wPi, \wf)$ is an admissible model. 
\end{prop}
\begin{proof}
If $\wDelta$ satisfies the upper-triangular property, then one can follow the same argument in \cite{Hai14a, HQ15} to verify that the pair $(\wPi, \wf)$ defined in \eqref{eq:wick_model} indeed satisfies all the requirements for an admissible model. It then remains to verify \eqref{eq:upper-triangular} for every basis vector $\tau$. 

The case for $\tau = \Xi$ and $X^{k}$ is trivial. One can also verify by hand that $\wDelta \Psi^{k} = \wM_\mu \Psi^{k} \otimes 1$, so the property also follows for $\tau = \Psi^{k}$. Now, since $\wM_\mu$ commutes with $\iI$ and $\eE^{k}$ as well as satisfies the rule \eqref{eq:leibniz}, it has exactly the same algebraic structures as in \cite{HQ15, HaiXu}, so it follows from the same line of argument that the upper-triangular property is preserved under the operations
\begin{align*}
\tau \mapsto \iI(\tau), \qquad \tau \mapsto \eE^{k} (\tau), \qquad \tau \mapsto X^{k} \tau. 
\end{align*}
It then extends to all basis vectors, and the claim follows. 
\end{proof}


We now turn to the description of the ``mass renormalisation" map $M_{0}$. Define the generators $L_{k,\ell}$ for $(k,\ell)$ both even or odd by
\begin{align*}
L_{2k,2\ell}: \qquad &\eE^{k-1} \big( \Psi^{2k} \iI (\eE^{\ell-1} \Psi^{2\ell}) \big) \mapsto \1, \\
&\eE^{k-1} \big( \Psi^{2k} \iI (\eE^{\ell-1} \Psi^{2\ell+1}) \big) \mapsto (2\ell+1) \cdot \Psi,
\end{align*}
\begin{align*}
L_{2k-1,2\ell+1}: \qquad &\eE^{k-1} \big( \Psi^{2k-1} \iI (\eE^{\ell-1} \Psi^{2\ell+1}) \big) \mapsto  \1, \\
&\eE^{k-1} \big( \Psi^{2k} \iI (\eE^{\ell-1} \Psi^{2\ell+1}) \big) \mapsto (2k) \cdot \Psi. 
\end{align*}
Given these generators, we then consider the map $M_{0}$ of the form
\begin{align*}
M_{0} = \exp \bigg( - \sum_{k,\ell \geq 1} C_{k,\ell} L_{k,\ell} \bigg). 
\end{align*}
$M_{0}$ is then parametrised by the set of constants $\{C_{k,\ell}\}$. For any admissible model $(\bar{\Pi}, \bar{f})$, we define the action of $M_{0}$ by
\begin{equation} \label{eq:mass_model}
\bar{\Pi}_{z}^{M_{0}} \tau := \bar{\Pi}_{z} M_{0} \tau, \qquad \bar{f}_{z}^{M_{0}} (\tau) := \bar{f}_{z}(\tau).  
\end{equation}
Since the only basis elements in the regularity structures that $M_{0}$ has a non-trivial effect on are of the form $\eE^{a} \Psi^{k} \big( \iI \big( \eE^{b} \Psi^{\ell} \big) \big)$, it is immediate to see that $M_{0}$ commutes with the elements in the transition group in the sense that $M_{0} \Gamma \tau = \Gamma M_{0} \tau$ for any $\Gamma \in \gG$ and $\tau \in \tT$. As a consequence, one could easily deduce that the model $(\bar{\Pi}^{M_{0}}, \bar{f}^{M_{0}})$ also belongs to $\sM$. 

Combining \eqref{eq:wick_model} and \eqref{eq:mass_model}, we can then define the new model $(\Pi^{M}, f^{M})$ under the action of $M = (\wM_{\mu}, M_{0})$ by
\begin{equation} \label{eq:renormalised_model}
\Pi^{M}_{z} \tau = (\Pi_{z} \otimes f_{z}) \wDelta (M_{0} \tau), \qquad f^{M}(\sigma) = f_{z} (\widehat{M}^{\text{Wick}} \sigma). 
\end{equation}
It follows from Proposition \ref{pr:upper-triangular} and the arguments right after \eqref{eq:mass_model} that $(\Pi^{M}, f^{M})$ is admissible as long as $(\Pi, f)$ is. With the above definitions of the renormalisation maps, we then have the following theorem. 

\begin{thm} \label{th:renormalised_eq}
	Let $\phi_{0} \in \cC^{1}, \epsilon \geq 0$ and $\zeta$ be a smooth space-time function. Let $(\Pi, f) = \sL_{\epsilon} (\zeta)$ be the canonical model, $M = (\wM_{\mu}, M_{0})$ be the renormalisation maps defined as above, and $(\Pi^{M}, f^{M}) = M \sL_{\epsilon} (\zeta)$ be the renormalised model as in \eqref{eq:renormalised_model}. If $\Phi \in \dD^{\gamma,\eta}_{\epsilon}$ is the fixed point solution to \eqref{eq:abstract} with the canonical model $\sL_{\epsilon}(\zeta)$, and $\rR^{M}$ is the associated map with the renormalised model, then the function $u = \rR^{M} \Phi$ solves the classical PDE
	\begin{align*}
	\partial_{t} u = \Delta u - \sum_{j=3}^{m} \lambda_{j} \epsilon^{j-1} W_{2j+1, \mu}(u) - \lambda_{0} u - C u + \zeta
	\end{align*}
	with initial data $\phi_{0}$, where
	\begin{align*}
	C = \sum_{k, \ell = 1}^{m} \lambda_{k} \lambda_{\ell} \bigg( (2k+1)(2k) C_{2k-1,2\ell+1} + (2k+1)(2\ell+1) C_{2k,2\ell} \bigg). 
	\end{align*}
\end{thm}
\begin{proof}
	The key of the proof is to note that $W_{n,\mu}$ defined in \eqref{eq:gen_Wick} is an Appell sequence, namely
	\begin{align*}
	W_{n,\mu}'(x) = n W_{n-1,\mu}(x). 
	\end{align*}
Of the the characteristics of  Appell sequence is  the following identity
	\begin{equation} \label{eq:polynomial}
	W_{n,\mu}(x+y) = \sum_{k=0}^{n} \begin{pmatrix} n \\ k \end{pmatrix} W_{k,\mu}(x) \cdot y^{n-k}
	\end{equation}
    as in the case of Hermite polynomial. Once \eqref{eq:polynomial} is established, the rest of the proof follows exactly the same way as \cite[Section 5.4]{HQ15} since the renormalisation map $M$ defined here has the same structure as in \cite{HQ15}. 
\end{proof}

We have now shown that the group of transformations we build does map $\sM$ to $\sM$, and have derived the form of the modified PDEs under renormalised models. The rest of the article is devoted to the proof of the convergence of these renormalised models (with suitably chosen renormalisation constants) and identification of the limit of the modified equations.

\section{Bounds on the renormalised models} \label{sec:bounds}

The main goal of this section to prove the bounds for high moments of the objects $|(\hPi_{z}^{\epsilon} \tau)(\varphi_{z}^{\lambda})|$, where $\hPi^{\epsilon}$ is the renormalised model as introduced in Section \ref{sec:renormalisation} with suitable measure $\mu_{\epsilon}$ and constants $C_{k,\ell}$'s. 

Let $\zeta$ be the random field satisfying Assumption \ref{ass:noise}, and $\zeta_{\epsilon}(t,x) = \epsilon^{-\frac{5}{2}} \zeta(t/\epsilon^{2}, x/\epsilon)$. Also, fix a space-time mollifier $\rho$ and for any $\bar{\epsilon} > 0$, set $\rho_{\bar{\epsilon}}(t,x) = \bar{\epsilon}^{-5} \rho (t/\bar{\epsilon}^{2}, x/\bar{\epsilon})$. Let
\begin{equation} \label{eq:smooth_noise}
\zeta_{\epsilon, \bar{\epsilon}} := \zeta_{\epsilon} * \rho_{\bar{\epsilon}}. 
\end{equation}
We can thus write $K * \zeta_{\epsilon, \bar{\epsilon}} = (K * \rho_{\bar{\epsilon}}) * \zeta_{\epsilon}$, and the kernel $K_{\bar{\epsilon}} := K * \rho_{\bar{\epsilon}}$ approximates $K$ with the bounds
\begin{equation} \label{eq:kernel_difference}
|K_{\bar{\epsilon}}(z) - K(z)| \lesssim \bar{\epsilon}^{\delta} |z|^{-3-\delta}, \qquad |DK_{\bar{\epsilon}}(z) - DK(z)| \lesssim \bar{\epsilon}^{\delta} |z|^{-4-\delta}
\end{equation}
for all sufficiently small $\delta$, uniformly over $\bar{\epsilon} < 1$ and $|z| < 1$. Later, we will use this bound to compare the difference between models built from $\zeta_{\epsilon}$ and $\zeta_{\epsilon,\bar{\epsilon}}$. 

Let $\mu_{\epsilon}$ and $\mu_{\epsilon, \bar{\epsilon}}$ denote the distributions of stationary solutions $\Psi_{\epsilon}$, $\Psi_{\epsilon, \bar{\epsilon}}$ to the equations
\begin{align*}
\partial_{t} \Psi_{\epsilon} = \Delta \Psi_{\epsilon} + \zeta_{\epsilon}, \qquad \partial_{t} \Psi_{\epsilon, \bar{\epsilon}} = \Delta \Psi_{\epsilon, \bar{\epsilon}} + \zeta_{\epsilon, \bar{\epsilon}}. 
\end{align*}
Let $M_{\epsilon} = (\wM_{\mu_{\epsilon}}, M_{0}^{(\epsilon)})$ and $M_{\epsilon, \bar{\epsilon}} = (\wM_{\mu_{\epsilon, \bar{\epsilon}}}, M_{0}^{(\epsilon, \bar{\epsilon})})$ be the renormalisation maps built in Section \ref{sec:renormalisation}, where the constants $C_{k,\ell}^{(\epsilon)}$ and $C_{k,\ell}^{(\epsilon,\bar{\epsilon})}$'s are chosen as in Section \ref{sec:values} below. The main theorem of this section is then the following. 

\begin{thm} \label{th:main_bound}
Let $\fM_{\epsilon} = (\hPi^{\epsilon}, \widehat{f}^{\epsilon}) = M_{\epsilon}\sL_{\epsilon}(\zeta_{\epsilon})$ and $\fM_{\epsilon, \bar{\eps}} = (\hPi^{\epsilon,\bar{\epsilon}}, \widehat{f}^{\epsilon,\bar{\epsilon}}) = M_{\epsilon, \bar{\eps}} \sL_{\epsilon}(\zeta_{\epsilon, \bar{\epsilon}})$, where $\sL_{\epsilon}(\cdot)$ is the canonical lift of the corresponding input to the space $\sM_{\epsilon}$, and the constants $C_{k,\ell}^{(\epsilon)}$ and $C_{k,\ell}^{(\epsilon, \bar{\epsilon})}$'s in defining the renormalisation groups are as set in Section \ref{sec:values} below. Then, there exists $\theta > 0$ such that for every $\tau \in \wW$ with $|\tau| < 0$ and every $p > 1$, we have
\begin{equation} \label{eq:main_bound}
\E |(\hPi_{z}^{\epsilon} \tau)(\varphi_{z}^{\lambda})|^{p} \lesssim \lambda^{p(|\tau| + \theta)}, \qquad \E |(\hPi_{z}^{\epsilon} \tau - \hPi_{z}^{\epsilon, \bar{\eps}}\tau)(\varphi_{z}^{\lambda})|^{p} \lesssim \bar{\eps}^{\theta} \lambda^{p(|\tau|+\theta)}, 
\end{equation}
where both bounds hold uniformly over all $\epsilon, \bar{\epsilon} \in (0,1)$, all test functions $\varphi \in \bB$ and all space-time points $z \in \R^{4}$. 
\end{thm}

We will mainly focus on the first bound in \eqref{eq:main_bound}, and will briefly discuss how the second bound follows from the first one via \eqref{eq:kernel_difference}. Also, since we are in a translation invariant setting, we will set $z=0$ without loss of generality.

\begin{rmk}
Ideally, one would like to find a limiting model $\hPi$ and prove a bound of the type
\begin{equation} \label{eq:ideal_bound}
\E |(\hPi_{z}^{\epsilon} \tau - \hPi_{z} \tau)(\varphi_{z}^{\lambda})|^{p} \lesssim \epsilon^{\theta} \lambda^{p(|\tau| + \theta)}. 
\end{equation}
for some positive $\theta$. In fact, the natural candidate for such a limiting model is the one in $\sM_{0}$ whose action on basis vectors without appearance of $\eE$'s coincides with the $\Phi^4_3$ model. In fact, we will actually prove this bound below for basis vectors that contain at least one appearance of $\eE$, in which case we have $\hPi_{z} \tau = 0$. We expect that the bound \eqref{eq:ideal_bound} still holds for standard $\Phi^4_3$ basis elements (those without appearance of $\eE$'s), but the proof would be much more involved due to the non-Gaussian noise. Later when we identify the $\Phi^4_3$ model as the limit, we make use of the convergence result available in the Gaussian case (\cite{HaiXu}) as well as the diagonal argument used \cite{KPZCLT} to circumvent the bound \eqref{eq:ideal_bound} in non-Gaussian case. 
\end{rmk}

\subsection{Graphic notation} \label{sec:graph}

Roughly speaking the random variables of the type $(\hPi_{0}^{\epsilon}\tau)(\varphi_{0}^{\lambda})$
in Theorem~\ref{th:main_bound} are all integrations of  convolutions or products of noises $\zeta_\eps$ and kernels $K$, $\varphi^\lambda$ and cumulant functions. We will apply
Definition~\ref{def:Wick-product} to rewrite a product of $\zeta_\eps$ as a sum of Wick products.
 Following \cite{HQ15, HaiXu, KPZCLT}, we introduce graphic notations to represent these integrals.

We denote by 
node \tikz[baseline=-3] \node [dot] {}; a space-time variable in $\R \times \R^{3}$ to be integrated out. The special green node \tikz[baseline=-3] \node [root] {}; denotes the origin $0$.
A bold green arrow 
\tikz[scale=0.7,baseline=-0.2] \draw[testfcn] (1,0.2) to (0,0.2); 
represents $\varphi^\lambda$, i.e. a generic test function rescaled by $\lambda$. 
Each plain arrow \tikz[scale=0.7,baseline=-0.2] \draw[kernel] (0,0.2) to (1,0.2); represents the kernel $K(z'-z)$, 
and a barred arrow \tikz[scale=0.7,baseline=-0.2] \draw[kernel1] (0,0.2) to (1,0.2); represents a factor $K(z'-z)-K(-z)$,
where $z$ and $z'$ are starting and end points of the arrow.
As for cumulants, we follow the notation in \cite{KPZCLT}: a gray ellipse with $p$ points inside
 \begin{tikzpicture}[scale=0.4,baseline=0] 
 \filldraw [lightgray] (0,0.3) ellipse (22pt and 18pt); 
 \node at (-0.3,0) [dot] {}; 
 \node at (-0.3,0.6) [dot] {}; 
 \node at (0.3,0) [dot] {}; 
 \node at (0.3,0.6) [dot] {}; 
 \end{tikzpicture}
 ($p=4$ here) represents the cumulant $\fC_{p}^{(\epsilon)}(z_{1}, \cdots, z_{p})$. 


To represent the Wick products we will need another type of special vertices $\tikz \node [var] {};$ in our graphs. Each instance of $\tikz \node [var] {};$ stands for an integration variable $x$, as well as a factor
$\zeta_\eps(x)$.
Furthermore, if more than one such vertex appear, then the corresponding product of 
$\zeta_\eps$ is always understood as a Wick product 
 $\Wick{ \zeta_\eps(x_1) \cdots \zeta_\eps(x_n)}$, where the $x_i$ are the integration variables represented by 
{\it all} of the special vertices $\tikz \node [var] {};$ appearing in the graph.

With these notations, for the canonical model 
$(\Pi^{\epsilon}, f^{\epsilon}) = \sL_{\epsilon}(\zeta_{\epsilon})$, we can apply \eqref{e:defWick} and write each $(\hPi_{0}^{\epsilon}\tau)(\varphi_{0}^{\lambda})$
as a sum of terms (``non-Gaussian chaos") each represented by a graph. For example, for $\tau = \eE \Psi^{4}$, we have
\begin{align*}
(\Pi_{0}^{\epsilon} \eE \Psi^{4})(\varphi_{0}^{\lambda}) \phantom{1} = \phantom{1} \epsilon \phantom{1} \bigg( \phantom{1}
\begin{tikzpicture} [scale=0.7,baseline=-9]
\node at (0,0) [dot] (center) {}; 
\node at (0,-1.2) [root] (below) {}; 
\node at (-1.2,0.9) [var] (farleft) {}; 
\node at (-0.4,1.3) [var] (left) {}; 
\node at (0.4,1.3) [var] (right) {}; 
\node at (1.2,0.9) [var] (farright) {}; 
\draw[kernel] (farleft) to (center); 
\draw[kernel] (left) to (center); 
\draw[kernel] (right) to (center); 
\draw[kernel] (farright) to (center); 
\draw[testfcn] (center) to (below); 
\end{tikzpicture}
+ \phantom{1} 6
\begin{tikzpicture} [scale=0.7,baseline=-9]
\filldraw[lightgray] (0,1.35) ellipse (13pt and 10pt); 
\node at (0,0) [dot] (center) {}; 
\node at (0,-1.2) [root] (below) {}; 
\node at (-1.2,0.9) [var] (farleft) {}; 
\node at (-0.15,1.4) [dot] (left) {}; 
\node at (0.15,1.4) [dot] (right) {}; 
\node at (1.2,0.9) [var] (farright) {}; 
\draw[kernel] (farleft) to (center); 
\draw[kernel, bend right=40] (left) to (center); 
\draw[kernel, bend left=40] (right) to (center); 
\draw[kernel] (farright) to (center); 
\draw[testfcn] (center) to (below); 
\end{tikzpicture}
+ \phantom{1} 6
\begin{tikzpicture} [scale=0.7,baseline=-9]
\filldraw[lightgray] (-0.95,1.1) ellipse (12pt and 12pt); 
\filldraw[lightgray] (0.95,1.1) ellipse (12pt and 12pt); 
\node at (0,0) [dot] (center) {}; 
\node at (0,-1.2) [root] (below) {}; 
\node at (-1,0.9) [dot] (farleft) {}; 
\node at (-0.9,1.3) [dot] (left) {}; 
\node at (0.9,1.3) [dot] (right) {}; 
\node at (1,0.9) [dot] (farright) {}; 
\draw[kernel,bend right=30] (farleft) to (center); 
\draw[kernel,bend left=30] (left) to (center); 
\draw[kernel,bend right=30] (right) to (center); 
\draw[kernel,bend left=30] (farright) to (center); 
\draw[testfcn] (center) to (below); 
\end{tikzpicture}
+ \phantom{1}
\begin{tikzpicture} [scale=0.7,baseline=-9]
\filldraw[lightgray] (0,1.2) ellipse (18pt and 12pt); 
\node at (0,0) [dot] (center) {}; 
\node at (0,-1.2) [root] (below) {}; 
\node at (-0.25,1.4) [dot] (farleft) {}; 
\node at (-0.25,1) [dot] (left) {}; 
\node at (0.25,1) [dot] (right) {}; 
\node at (0.25,1.4) [dot] (farright) {}; 
\draw[kernel, bend right = 80] (farleft) to (center); 
\draw[kernel] (left) to (center); 
\draw[kernel] (right) to (center); 
\draw[kernel, bend left = 80] (farright) to (center); 
\draw[testfcn] (center) to (below); 
\end{tikzpicture} \phantom{1} \bigg), 
\end{align*}
where the first and second terms above have the full expressions
\begin{align*}
I_{1} = &\int \varphi_{0}^{\lambda}(z) \prod_{j=1}^{4} K(z-x_{j}) \Wick{\prod_{j=1}^{4} \zeta_{\epsilon}(x_{j})} dx dz, \\
I_{2} = 6 &\int \varphi_{0}^{\lambda}(z) \prod_{j=1}^{4} K(z-x_{j}) \cdot \fC_{2}^{(\epsilon)}(x_{1},x_{2}) \Wick{\zeta_{\epsilon}(x_{3}) \zeta_{\epsilon}(x_{4})} dx dz, 
\end{align*}
and we used $x = (x_{1}, \cdots, x_{4})$. If $\hPi^{\epsilon}$ is the renormalised model, then $(\hPi^{\epsilon}_{0} \tau)(\varphi_{0}^{\lambda})$ consists only the first term above. Similarly, for $\tau = \Psi^{2} \iI(\Psi^{2})$, we have
\begin{align*}
(\Pi_{0}^{\epsilon} \Psi^{2} \iI(\Psi^{2}))(\varphi_{0}^{\lambda}) \phantom{1} &= 
\begin{tikzpicture} [scale=0.7,baseline=-9]
\node at (-0.8,1) [var] (left) {}; 
\node at (0.8,1) [var] (right) {}; 
\node at (0,0) [dot] (center) {}; 
\node at (0,-1.3) [dot] (below) {}; 
\node at (0.8,-0.3) [var] (belowright) {}; 
\node at (-0.8,-0.3) [var] (belowleft) {}; 
\node at (-1.5,-1.5) [root] (zero) {}; 
\draw[kernel] (left) to (center); 
\draw[kernel] (right) to (center); 
\draw[kernel] (belowright) to (below); 
\draw[kernel] (belowleft) to (below); 
\draw[kernel1] (center) to (below); 
\draw[testfcn] (below) to (zero); 
\end{tikzpicture}
\;\;+ \phantom{1} 4
\begin{tikzpicture} [scale=0.7,baseline=-9]
\node at (-0.8,-0.3) [var] (belowleft) {}; 
\node at (-0.8,1) [var] (left) {}; 
\node at (0,0) [dot] (center) {}; 
\node at (0,-1.3) [dot] (below) {}; 
\node at (-1.5,-1.5) [root] (zero) {}; 
\draw[kernel] (left) to (center); 
\draw[kernel] (belowleft) to (below); 
\draw[kernel1] (center) to (below); 
\draw[testfcn] (below) to (zero); 
\filldraw[lightgray] (1,-0.65) ellipse (10pt and 15pt); 
\draw[kernel] (1,-0.4) to (center); 
\draw[kernel] (1,-0.9) to (below); 
\node at (1,-0.4) [dot] {}; 
\node at (1,-0.9) [dot] {}; 
\end{tikzpicture}
\;\;+ \phantom{1} 
\begin{tikzpicture} [scale=0.7,baseline=-9]
\filldraw[lightgray] (1.2,-1.3) ellipse (10pt and 15pt); 
\node at (0.8,1) [var] (right) {}; 
\node at (-0.8,1) [var] (left) {}; 
\node at (0,0) [dot] (center) {}; 
\node at (0,-1.3) [dot] (below) {}; 
\node at (-1.5,-1.5) [root] (zero) {}; 
\draw[kernel] (left) to (center); 
\draw[kernel] (right) to (center); 
\draw[kernel1] (center) to (below); 
\draw[testfcn] (below) to (zero); 
\node at (1.2,-1.1) [dot] (one) {}; 
\node at (1.2,-1.5) [dot] (two) {}; 
\draw[kernel, bend right=50] (one) to (below); 
\draw[kernel, bend left=50] (two) to (below); 
\end{tikzpicture}\\
&\phantom{11}+ 2 \phantom{1}
\begin{tikzpicture} [scale=0.7,baseline=-9]
\node at (0,1) [dot] (up) {}; 
\node at (0,-1) [dot] (below) {}; 
\node at (-1.5,-1.3) [root] (zero) {}; 
\draw[kernel1] (up) to (below); 
\draw[testfcn] (below) to (zero); 
\filldraw[lightgray] (-1,0) ellipse (10pt and 15pt); 
\filldraw[lightgray] (1,0) ellipse (10pt and 15pt); 
\node at (-1,0.2) [dot] (leftup) {}; 
\node at (-1,-0.2) [dot] (leftbelow) {}; 
\node at (1,0.2) [dot] (rightup) {}; 
\node at (1,-0.2) [dot] (rightbelow) {}; 
\draw[kernel] (up) to (leftup); 
\draw[kernel] (up) to (rightup); 
\draw[kernel] (below) to (leftbelow); 
\draw[kernel] (below) to (rightbelow); 
\end{tikzpicture}
\phantom{1} +
\begin{tikzpicture} [scale=0.7,baseline=-9]
\node at (0,1) [dot] (up) {}; 
\node at (0,-1) [dot] (below) {}; 
\node at (-1.5,-1.3) [root] (zero) {}; 
\draw[kernel1] (up) to (below); 
\draw[testfcn] (below) to (zero); 
\filldraw[lightgray] (1,0) ellipse (14pt and 18pt); 
\node at (0.8,0.2) [dot] (leftup) {}; 
\node at (0.8,-0.2) [dot] (leftbelow) {}; 
\node at (1.2,0.2) [dot] (rightup) {}; 
\node at (1.2,-0.2) [dot] (rightbelow) {}; 
\draw[kernel] (up) to (leftup); 
\draw[kernel] (below) to (leftbelow); 
\draw[kernel, bend left=50] (up) to (rightup); 
\draw[kernel, bend right=50] (below) to (rightbelow); 
\end{tikzpicture}. 
\end{align*}
The reason that there is no contraction for odd number of vertices is because the noise is symmetric so all its odd cumulants vanish, and there is no contraction between two top vertices because the kernel $K$ chosen in Section~\ref{sec:rs} annihilates constants. 

In order to bound the moments of the random variables $(\hPi_{0}^{\epsilon} \tau)(\varphi_{0}^{\lambda})$, we bound the moments of the term represented by {\it each} graph.
Assumption~\ref{ass:ele-graph} and Theorem~\ref{th:moment-bound} below will allow us to conclude the desired moment bounds by simply verifying certain graphical conditions. To state these conditions we introduce the following terminologies.

Let $(H,\CE)$ be a graph
where $H$ is the set of vertices and $\eE$ (with an abuse use of notation as this letter also stands for the operator as multiplying by $\epsilon$) is the collection of edges. Each such graph consists of a set $H_{ex}$ of external vertices (i.e. noise vertices that have not been integrated out) and a set $H_{in}$ of internal vertices and the origin $0$, namely $H=\{0\} \cup H_{ex}\cup H_{in}$. There is a distinguished vertex $v_\star$ connected with $0$ by a line representing the test function $\varphi^\lambda$.
Write 
\begin{equ}
H_0=H\setminus \{0\} \qquad \mbox{and} \qquad H_\star=\{0,v_\star\} \;.
\end{equ}
For a sub-graph $\bar H \subset H$, we write $\bar H_{in}=\bar H \cap H_{in}$
and $\bar H_{ex}=\bar H \cap H_{ex}$.

The set of edges is decomposed as $\CE=\CE_2 \cup \CE_h$ 
where $\CE_2$ is a set of usual directed edges, which represent integration kernels, 
and $\CE_h$ is a set of hyper-edges (i.e. edges consisting of more than two vertices), which represent cumulants of order higher than two. Each edge $e = (x_{v_{-}}, x_{v_{+}}) \in \eE_{2}$ comes with a pair of numbers $(a_{e}, r_{e}) \in \R^{+} \times \R$. The number $a_{e} > 0$ measures the singularity of the kernel $J_{e}$ associated to the edge $e$. More precisely, for each edge $e$ with $a_{e}$, the kernel $J_{e}$ associated to it satisfies the bound
\begin{align*}
|D^{k}J_{e} (z)| \lesssim |z|^{-a_{e} - |k|}. 
\end{align*}
If an edge $e \in \eE_{2}$ is connected to an external vertex (noise), then the kernel $J_{e}$ has the bound
\begin{align*}
|D^{k}J_{e} (z)| \lesssim (|z| + \epsilon)^{-a_{e} - |k|}, 
\end{align*}
reflecting the fact that the input of the noise is at scale $\epsilon$. For every $\beta > 0$, we associate the kernel $J_{e}$ the norm $\| J_{e} \|_{a_{e};\beta}$ by
\begin{align*}
\| J_{e} \|_{a_{e};\beta} := \sup_{|z|\leq 1, |k|\leq \beta} |D^{k}J_{e}(z)|. 
\end{align*}
This quantity is always finite by assumption on the kernel $J_{e}$. The number $r_{e}$, on the other hand, gives the corresponding renormalisations needed for the edge: if $r_{e} > 0$, then one needs to subtract Taylor expansions up to order $r_{e}$; if $r_{e} < 0$, then one needs to define a renormalised version of the kernel by subtracting Taylor polynomials of the test function at the origin up to degree $|r_{e}|$. The orientation of an edge matters only when $r_{e} > 0$. Precise definitions of the positive/negative renormalisations could be found in \cite[Section 8]{HQ15} and \cite[Section 4]{HaiXu}. 

For each hyper-edge $e \in \eE_{h}$, we set $r_{e} = 0$, and associate it with the degree $a_{e} = \frac{5n}{2}$, where $n$ is the number of vertices in this hyper-edge. We use the same notation
 \begin{tikzpicture}[scale=0.4,baseline=0] 
 \filldraw [lightgray] (0,0.3) ellipse (22pt and 18pt); 
 \node at (-0.3,0) [dot] {}; 
 \node at (-0.3,0.6) [dot] {}; 
 \node at (0.3,0) [dot] {}; 
 \node at (0.3,0.6) [dot] {}; 
 \end{tikzpicture}
to represent a hyper-edge connecting $n$ noise vertices, and it represents the kernel $\fC_{\epsilon} (z_{1}, \cdots, z_{n})$. Note that this kernel has the correct scaling behaviour corresponding to the degree $\frac{5n}{2}$, and satisfies  \eqref{eq:cumu-close}. 

\begin{rmk}
	The degree of the hyper-edge do correspond to the correct behaviour of cumulants of the rescaled field $\zeta_{\epsilon}$. This is guaranteed by the scale set in \eqref{eq:rescaled_field} as well as \eqref{eq:cumu-close}. For normal edges $e \in \eE_{2}$, we assign their initial degree to be $a_{e} = 3$. But many of these (for those connected to external noise vertices) will be reduced later by multiplication of powers of $\epsilon$'s. All edges represented by the plain arrow \tikz[baseline=-0.2] \draw[kernel] (0,0.15) to (1,0.15); have $r_{e} = 0$, but the pair corresponding to the barred arrow \tikz[baseline=-0.2] \draw[kernel1] (0,0.15) to (1,0.15); is always $(a_{e},r_{e}) = (3,1)$. 
\end{rmk}

For a sub-graph $\bar H \subset H$, we define the sets
$\CE^\uparrow(\bar H)$, $\CE^\downarrow(\bar H)$,
$\CE_0(\bar H)$ and $\CE(\bar H)$ in the same way as \cite{HQ15, HaiXu} by
\begin{align*}
\eE^{\uparrow}(\bar H) &= \{e\in\eE: \eE \cap \bar H = e_{-}, r_{e} > 0 \}; \\
\eE^{\downarrow}(\bar H) &= \{e \in \eE: \eE \cap \bar H = e_{+}, r_{e} > 0\}; \\
\eE_{0}(\bar H) &= \{ e \in \eE: \eE \cap \bar H = e \}; \\
\eE(\bar H) &= \{ e \in \eE: \eE \cap \bar H \neq \emptyset \}. 
\end{align*}
In particular $\CE^\uparrow(\bar H)$ and $\CE^\downarrow(\bar H)$ are subsets of $\CE_2$, and a hyper-edge $e\in\CE_0(\bar H)$ only if {\it all} the vertices of $e$ are in $\bar H$. We make a few assumptions on our graphs that ensures the high moments of the corresponding object have the correct scaling behaviour. 

\begin{assumption}\label{ass:ele-graph}
	The labelled graph $(H,\CE)$ satisfies the following properties.
	\begin{itemize}\itemsep0em
		\item[1.] For every edge $e \in \CE_2$, one has $a_e + (r_{e} \wedge 0) < 5$.
		\item[2.] For every subset
		$\bar H \subset H_0$ of cardinality at least $3$, one has 
		\begin{equ}[e:aEdges]
			\sum_{e \in \CE_0(\bar H)} a_e 
			< 5\,\Big(|\bar H_{in}| +{1\over 2} (|\bar H_{ex}| -1 -\one_{\bar H_{ex}=\emptyset}) \Big) \;.
		\end{equ}
		\item[3.] For every subset
		$\bar H \subset H$ containing $0$ and of cardinality at least $2$, one has 
		\begin{equ}[e:aEdges1]
			\sum_{e \in \CE_0(\bar H)} a_e 
			+ \sum_{e \in \CE^\uparrow( \bar H)}(a_e + r_e - 1)
			- \sum_{e \in \CE^\downarrow(\bar H)} r_e 
			< 5\,\Big(|\bar H_{in} | + \frac{1}{2}| \bar H_{ex} |\Big).
		\end{equ}
		\item[4.] For every non-empty subset $\bar H \subset H_\star$,
		one has the bounds
		\begin{equ}[e:aEdges2]
			\sum_{e \in \CE(\bar H) \setminus \CE^\downarrow(\bar H)} a_e  
			+ \sum_{e \in \CE^\uparrow(\bar H)} r_e
			- \sum_{e \in \CE^\downarrow(\bar H)} (r_e-1) 
			> 5\,\Big(|\bar H_{in}| +{1\over 2} |\bar H_{ex}| \Big) \;.
		\end{equ}
	\end{itemize}
\end{assumption}

\begin{rmk}
	The $p$-th moment of the object represented by a graph could be expressed as a sum of finitely many terms, each obtained by Wick-contracting $p$ copies of the graph (that is, excluding self-contraction). We refer to \cite{KPZCLT} for more details on the Wick-contraction. 
\end{rmk}

The following theorem from \cite{AjayHao} gives a sufficient condition on the graph $H$ for its $p$-th moment to have the correct scaling behaviour. 

\begin{thm} \label{th:moment-bound}
	Suppose that a graph $(H, \eE)$ satisfies Assumption \ref{ass:ele-graph}, and the kernels $\fC_{\epsilon} (z_{1}, \cdots, z_{n})$ represented by the hyper-edges are the cumulants of the rescaled noise satisfying Assumption \ref{ass:noise}. If $I_{p,\lambda}^{H}$ denotes the $p$-th moment of the (random) object represented by $H$ with test function scaled at $\lambda$, then there exists $\beta > 0$ depending on the structure of the graph only such that
	\begin{align*}
	I_{p, \lambda}^{H} \lesssim \lambda^{\alpha p} \cdot \prod_{e} \|J_{e}\|_{a_{e};\beta}^{p}, \quad \alpha = 5 |H_{\text{in}} \setminus H_{\star}| + \frac{5}{2}|H_{\text{ex}}| - \sum_{e} a_{e}, 
	\end{align*}
	where the proportionality constant depends on the structure of graph and cumulants of the unscaled noise $\zeta(\cdot)$ only. 
\end{thm}

\subsection{Values of renormalisation constants} \label{sec:values}

With the graphic notations, we now give explicit expressions to the constants $C_{k,\ell}^{(\epsilon)}$'s that appear in the renormalisation map $M_{0}^{(\epsilon)}$ for $k, \ell$ both even and both odd. Given a tuple of two integers $(k,\ell)$, 
we say that $\pi$ is a {\it pairing} of $(k,\ell)$ if $\pi$ is a collection of $n$ tuples (each consisting of two integers):
\begin{equ} \label{e:partition-form}
\pi = \big\{ (k_{1}, \ell_{1}), \cdots, (k_{n},\ell_{n}) \big\}
\end{equ}
for some $n$ such that $k_{j}, \ell_{j} \geq 1$ for all $j$ and
\begin{equ}
\sum_{j} k_{j} = k, \qquad \sum_{j} \ell_{j} = \ell. 
\end{equ}
We write $|\pi|=n$. Let $\pP(k,\ell)$ be the set of all the pairings of $(k,\ell)$.
For example, we have 
\begin{align*}
\pP(3,3) = \bigg\{ \big\{ (1,1),(1,1),(1,1) \big\}, \big\{ (1,1),(2,2) \big\}, \big\{ (1,2),(2,1) \big\},\big\{ (3),(3) \big\} \bigg\}. 
\end{align*}
Given a pairing $\pi \in \pP(k,\ell)$, we define
\begin{align*}
C_{k,\ell,\pi}^{(\epsilon)} \quad = \quad \epsilon^{\frac{k+\ell-4}{2}} \phantom{1}
	\begin{tikzpicture}[scale=0.8,baseline=-0.0cm]
	\node at (-1.2,0) [dot] (left){};
	\node at (1.2,0) [root] (right) {}; 
	\node[cloud, cloud puffs=7.7, cloud ignores aspect, minimum width=1.5cm, minimum height=1.3cm,  draw=lightgray, fill=lightgray]  at (0,1.5) {};
	%
	\draw[kernel,bend right=50] (left) to (right); 
	\draw[kernel,bend left=60] (-0.1,1) to (left); 
	\draw[kernel,bend right=60] (-0.3,1.3) to (left); 
	\draw[kernel,bend right=60] (0.1,1) to (right); 
	\draw[kernel,bend left=60] (0.3,1.3) to (right); 
	\node at (-0.8,0.4) {$\dots$}; 
	\node at (0.8,0.4) {$\dots$};
	\node at (-0.8,0.7) {\scriptsize $(k)$};
	\node at (0.8,0.7) {\scriptsize $(\ell)$};
	\node at (0,1.5) {$\pi$};
	\node at (-0.1,1) [dot] {}; 
	\node at (0.1,1) [dot] {}; 
	\node at (-0.3,1.3) [dot] {}; 
	\node at (0.3,1.3) [dot] {}; 
	\end{tikzpicture}
\end{align*}
In the above picture  the cloudy area represents a product of $|\pi|$ cumulants specified by $\pi$: if $\pi$ is given by \eqref{e:partition-form}
then the $i$-th  cumulant ($1\le i\le n$) in the product is a cumulant function
of order $k_i+\ell_i$ with $k_i$ variables selected from the $k$ ones on the left
and $\ell_i$ variables selected from the $\ell$ ones on the right.
Note that there is a slight abuse of notation here as $\pi$ itself depends on $k$ and $\ell$, but we choose to keep this notation for simplicity. 
Intuitively, a pairing $\pi$ is a way to ``contract" the $k$ vertices on the left
and the $\ell$ ones on the right.

As for the values of these constants, it is easy to see that $C_{2,2,\pi}^{(\epsilon)}$ diverges logarithmically when $\pi=\{(1,1),(1,1)\}$, while all other $C_{k,\ell,\pi}^{(\epsilon)}$'s converge to a finite limit. In fact, we have that if $\pi=\{(1,1),(1,1)\}$ then
\begin{equation} \label{eq:log_constant}
C_{2,2,\pi}^{(\epsilon)} \phantom{1} = 
\begin{tikzpicture}[scale=0.8,baseline=-0.0cm]
\node at (-1.2,0) [dot] (left){};
\node at (1.2,0)   [root] (right) {}; 
\filldraw [lightgray] (0,1.4) ellipse  (10pt and 5pt); 
\filldraw [lightgray] (0,0.9) ellipse  (10pt and 5pt); 
\draw[kernel,bend right=50] (left) to (right); 
\draw[kernel,bend left=50] (-0.15,0.9) to (left); 
\draw[kernel,bend right=50] (-0.15,1.4) to (left); 
\draw[kernel,bend right=50] (0.15,0.9) to (right); 
\draw[kernel,bend left=50] (0.15,1.4) to (right); 
\node at (-0.15,0.9) [dot] {}; 
\node at (0.15,0.9) [dot] {}; 
\node at (-0.15,1.4) [dot] {}; 
\node at (0.15,1.4) [dot] {}; 
\end{tikzpicture}
= C_{\log} \cdot |\log \epsilon| + \oO(1), 
\end{equation}
for some universal constant $C_{\log}$, while for $\pi'$ being the single contraction of all four points together, we have
\begin{equation} \label{eq:finite_constant}
C_{2,2,\pi'}^{(\epsilon)} \phantom{1} = 
\begin{tikzpicture}[scale=0.8,baseline=-0.0cm]
\node at (-1.2,0) [dot] (left){};
\node at (1.2,0)   [root] (right) {}; 
\filldraw [lightgray] (0,1.25) ellipse (15pt and 10pt); 
\draw[kernel,bend right=50] (left) to (right); 
\draw[kernel,bend left=50] (-0.15,1.15) to (left); 
\draw[kernel,bend right=50] (-0.15,1.44) to (left); 
\draw[kernel,bend right=50] (0.15,1.15) to (right); 
\draw[kernel,bend left=50] (0.15,1.45) to (right); 
\node at (-0.15,1.15) [dot] {}; 
\node at (0.15,1.15) [dot] {}; 
\node at (-0.15,1.45) [dot] {}; 
\node at (0.15,1.45) [dot] {}; 
\end{tikzpicture}
= C_{2,2,\pi'} + \oO(\epsilon). 
\end{equation}
For fixed $k, \ell \geq 2$ and $\pi \in \pP(k, \ell)$, we let $\pi !$ denote all the ways to contract $(k, \ell)$ vertices according to the pairing $\pi$. For example, for $(k, \ell) = (3,3)$, we have 
\begin{align*}
\pi! = 6, \qquad &\pi = \{(1,1), (1,1), (1,1)\}; \\
\pi! = 9, \qquad &\pi = \{(1,1), (2,2)\}; \\
\pi! = 9, \qquad &\pi = \{(1,2), (2,1)\} ; \\ 
\pi! = 1, \qquad &\pi = \{(3), (3)\}.
\end{align*}
With this notation, we let
\begin{equation} \label{eq:sum_constants}
C_{k,\ell}^{(\epsilon)} = \sum_{\pi \in \pP(k,\ell)} \pi! \cdot C_{k, \ell, \pi}^{(\epsilon)}, 
\end{equation}
where the sum is taken over all pairings for $(k,\ell)$. The precise values of these constants do not matter, so we do not give explicit formulae of values of $\pi !$ for general partitions $\pi$. 

As for $C_{k,\ell,\pi}^{(\epsilon, \bar{\epsilon})}$'s, they are almost the same as the $C_{k,\ell,\pi}^{(\epsilon)}$'s except that each plain arrow \tikz[baseline=-0.2] \draw[kernel] (0,0.1) to (1,0.1); connected to a noise is replaced by a dashed arrow \tikz[baseline=-0.2] \draw[kepsilon] (0,0.1) to (1,0.1);, representing a factor $K * \rho_{\bar{\epsilon}}$. This is because we have $K * \zeta_{\epsilon,\bar{\epsilon}} = K * \rho_{\bar{\epsilon}} * \zeta_{\epsilon}$. Since we have put the mollifier $\rho_{\bar{\epsilon}}$ into the kernel $K$, the notion
\begin{tikzpicture}[scale=0.4,baseline=0] 
\filldraw [lightgray] (0,0.3) ellipse (22pt and 18pt); 
\node at (-0.3,0) [dot] {}; 
\node at (-0.3,0.6) [dot] {}; 
\node at (0.3,0) [dot] {}; 
\node at (0.3,0.6) [dot] {}; 
\end{tikzpicture}
still represents the cumulants of the field $\zeta_{\epsilon}$. For example, we have
\begin{align*}
C_{k,\ell,\pi}^{(\epsilon, \bar{\epsilon})} \quad = \quad \epsilon^{\frac{k+\ell-4}{2}} \phantom{1}
\begin{tikzpicture}[scale=0.8,baseline=-0.0cm]
\node at (-1.2,0) [dot] (left){};
\node at (1.2,0) [root] (right) {}; 
\node[cloud, cloud puffs=7.7, cloud ignores aspect, minimum width=1.5cm, minimum height=1.3cm,  draw=lightgray, fill=lightgray]  at (0,1.5) {};
%
\draw[kernel,bend right=50] (left) to (right); 
\draw[kepsilon,bend left=60] (-0.1,1) to (left); 
\draw[kepsilon,bend right=60] (-0.3,1.3) to (left); 
\draw[kepsilon,bend right=60] (0.1,1) to (right); 
\draw[kepsilon,bend left=60] (0.3,1.3) to (right); 
\node at (-0.8,0.4) {$\dots$}; 
\node at (0.8,0.4) {$\dots$};
\node at (-0.8,0.7) {\scriptsize $(k)$};
\node at (0.8,0.7) {\scriptsize $(\ell)$};
\node at (0,1.5) {$\pi$};
\node at (-0.1,1) [dot] {}; 
\node at (0.1,1) [dot] {}; 
\node at (-0.3,1.3) [dot] {}; 
\node at (0.3,1.3) [dot] {}; 
\end{tikzpicture}. 
\end{align*}
The sum $C_{k,\ell}^{(\epsilon,\bar{\epsilon})}$'s are defined the same way as in \eqref{eq:sum_constants}. 

\subsection{First order renormalisation bounds}

We are now ready to prove Theorem \ref{th:main_bound} for all basis vectors $\tau \in \wW$ with $|\tau| < 0$, as listed in \eqref{e:simbols}. 
We first prove the bound \eqref{eq:main_bound} for $\tau = \eE^{k-1} \Psi^{2k+1-n}$ where $n \leq 3$. For the canonical model $\Pi^{\epsilon}$, then we have
\begin{align*}
	\Pi_{z}^{\epsilon} \tau = \epsilon^{k-1} (\Pi_{z}^{\epsilon} \Psi)^{2k+1-n}. 
\end{align*}
If $\Psi_{\epsilon} = P * \zeta_{\epsilon}$ is the stationary solution to the linear heat equation driven by $\zeta_{\epsilon}$, 
and $\mu_{\epsilon}$ is the distribution of $\Psi_{\epsilon}$, we have 
\begin{align*}
W_{k, \mu_{\epsilon}} (\Psi_{\epsilon}(z)) = \int P(z-x_{1}) \cdots P (z - x_{k}) \Wick{\zeta_{\epsilon}(x_{1}) \cdots \zeta_{\epsilon}(x_{k})} d x_{1} \cdots d x_{k}, 
\end{align*}
which implies
\vspace{-2em}
\begin{equation} \label{eq:first_object}
	(\hPi_{0}^{\epsilon} \tau)(\varphi_{0}^{\lambda})  =  (\Pi_{0}^{\epsilon} M_{\mu_{\epsilon}} \tau)(\varphi_{0}^{\lambda}) = \quad \epsilon^{k-1} \;
	\begin{tikzpicture}[scale=0.6,baseline=-0.5cm]
	\node at (-1,0) [var] (left){};
	\node at (1,0)  [var] (right) {};
	\node at (0,-1)   [dot] (middle) {}; 
	\node at (0, -2)  [root] (below) {}; 
	\node at (0,0) {$\cdots$}; 
	\draw[kernel] (left) to (middle); 
	\draw[kernel] (right) to (middle); 
	\draw[testfcn] (middle) to (below); 
	\draw [decorate,decoration={brace,amplitude=7pt}] (-1,0.3) to node[midway, yshift=0.5cm] {\tiny $2k+1-n$} (1,0.3); 
	\end{tikzpicture}
	\; . 
\end{equation}
Here, each plain arrow $\tikz[baseline=0.2] \draw[kernel] (0,0.15) to (1,0.15);$ has degree $a_{e} = 3$, which could be reduced by assigning suitable powers of $\epsilon$'s to it. We have the following proposition. 

\begin{prop} \label{pr:1st-order}
Both bounds in \eqref{eq:main_bound} hold for $\tau = \eE^{k-1} \Psi^{2k+1-n}$ with $k \geq 1$ and $n = 0, 1, 2, 3$. 
\end{prop}
\begin{proof}
	We first prove the first bound in \eqref{eq:main_bound}, and briefly discuss how the second one follows from it immediately. In order to make use of the positive homogeneity of $\eE$, we first assign powers of $\epsilon$'s to the edges of the graph in \eqref{eq:first_object}. We do it in the following way. Let $k \geq 1$ be given. Take any $2k-2$ edges connected to the noise from the $2k+1-n$ ones, and assign a $(\frac{1}{2}-\delta)$ power of $\epsilon$ to each of these edges. Thus, in the $2k+1-n$ noise edges, $3-n$ of them have degree $a_{e} = 3$, and all the rest $2k-2$ ones have degree $\frac{5}{2} + \delta$, and the graph comes with a multiplication of $\epsilon^{(2k-2)\delta}$. 
	
	Now, since $\tau$ has homogeneity $|\tau| = -\frac{1}{2}(3-n) - (2k+1-n) \kappa$, and the homogeneity (i.e. value of $\alpha$ in Theorem~\ref{th:moment-bound}) of its associated graph (after the allocation of $\epsilon$'s) is
	\begin{align*}
	\frac{5}{2}|H_{\text{ex}}| - \sum_{e \in \eE} a_{e}  
	&=\frac{5}{2}(2k+1-n)-3\,(3-n)-(\frac52 +\delta)\,(2k-2) \\
	&= - \frac{1}{2}(3-n) - (2k-2)\, \delta, 
	\end{align*}
	it follows that if we take $\delta$ small enough, then in view of Theorem \ref{th:moment-bound}, it suffices to check all the conditions in Assumption \ref{ass:ele-graph}. 
	
	The first condition is obvious. For the rest of the conditions, we take an arbitrary sub-graph $\bar{H}$ that contains $\ell$ noise vertices: $p$ of them have degree $3$, and the rest $\ell - p$ have degree $\frac{5}{2} + \delta$. By assumption, we know $p \leq 3$. 
	
	For \eqref{e:aEdges}, if $v_{\star} \in \bar{H}$, then the condition reads
	\begin{align*}
	3p + (\frac{5}{2} + \delta) (\ell - p) < 5 (1 + \frac{\ell}{2}), 
	\end{align*}
	which certainly holds since $p \leq 3$. If $v_{\star} \notin \bar{H}$, then the right hand side of \eqref{e:aEdges} becomes $0$, which trivially make the condition satisfied. 
	
	As for \eqref{e:aEdges1}, since the graph does not contain any edge with $r_{e} \neq 0$, its left hand side is identical to that of \eqref{e:aEdges}, while the right hand side is strictly larger. Thus, the condition is automatically implied by \eqref{e:aEdges} for this $\tau$. 
	
	We finally turn to \eqref{e:aEdges2}, where we only consider the case $v_{\star} \notin \bar{H}$. The RHS of \eqref{e:aEdges2} is $5 \ell /2$, while the LHS is a sum of $\ell$ elements, each being at least $\frac{5}{2} + \delta$, so the condition also holds. 
	
	By Theorem \ref{th:moment-bound}, we have already shown the first bound in \eqref{eq:main_bound}. 
	As for the second one, the expression turns out to be a sum of graphs of the same type, but in each term exactly one instance of $K$ is replaced by $K - K_{\bar{\epsilon}}$. The bound for the difference of the kernels together with the first bound in \eqref{eq:main_bound} immediately imply the second one. 
\end{proof}

\begin{rmk}
	Condition \eqref{e:aEdges} fails only when $p \geq 5$. Since we are in a regime where $p=3$, this reflects the sub-criticality of our equation. It also indicates that the equation becomes critical when the nonlinear term is $u_{\epsilon}^{5}$ but without any $\epsilon$'s in front of it. 
\end{rmk}

\begin{rmk}
	There is a multiple of $\epsilon^{(2k-2)\delta}$ of the graph after the $\epsilon$-allocation. This gives
	\begin{align*}
	\E |(\hPi_{0}^{\epsilon} \eE^{k-1} \Psi^{2k+1-n})(\varphi_{0}^{\lambda})|^{p} \lesssim \epsilon^{(2k-2)\delta p} \lambda^{p(|\tau| + \theta)}. 
	\end{align*}
	In the case $k \geq 2$, the power of $\epsilon$ is strictly positive, which corresponds to the stronger bound in \eqref{eq:ideal_bound} with $\hPi_{z} \tau = 0$. 
\end{rmk}

\subsection{Bounds for second order objects in higher homogeneous chaos} \label{sec:second_order_bounds}

The second order objects are homogeneous chaos decomposition of basis vectors of the form
\begin{equation} \label{eq:second_object}
\tau = \eE^{\lfl (k-1)/2 \rfl} \big( \Psi^{k} \iI(\eE^{\lfl \ell/2  \rfl - 1} \Psi^{\ell}) \big). 
\end{equation}
Note that by the expansion of the formal right hand side of the abstract equation \eqref{eq:abstract}, the only situation that will \textit{not} appear is when $k$ is odd and $\ell$ is even. The action of the Wick renormalised model 
on this object yields a sum of terms in different homogeneous chaos. Each of them can be expressed by a graph of the type
\begin{equation} \label{eq:graph}
	\epsilon^{\frac{1}{2}(p+q+n-4-\delta_{\tau})} \phantom{1}
	\begin{tikzpicture}[scale=0.7,baseline=-0.2cm]
	\node at (-1.5,0) [dot] (left){};
	\node at (-1.3,-0.5) {$v^\star$}; 
	\node at (1.5,0) [dot] (right) {}; 
	\node at (1.6,-0.7) {$v_{\star}$};  
	\node at (.5,-1.5) [root] (below) {}; 
	\node at (-2.4,1.2) [var] (aboveleft) {}; 
	\node at (-2.4,-1.2) [var] (belowleft) {}; 
	\node at (2.4,1.2) [var] (aboveright) {}; 
	\node at (2.4,-1.2) [var] (belowright) {}; 
	\node at (-2.4,-0.3) {$\vdots$}; 
	\node at (-2.4,0.7) {$\vdots$}; 
	\node at (2.4,-0.3) {$\vdots$}; 
	\node at (2.4,0.7) {$\vdots$}; 
	\node at (-1,0.8) {$\vdots$}; 
	\node at (1,0.8) {$\vdots$}; 
	\draw[kernel1,bend right=30] (left) to (right); 
	\draw[testfcn] (right) to (below); 
	\draw[kernel] (aboveleft) to (left); 
	\draw[kernel] (belowleft) to (left); 
	\draw[kernel] (aboveright) to (right); 
	\draw[kernel] (belowright) to (right); 
	\node[cloud, cloud puffs=7.7, cloud ignores aspect, minimum width=1.3cm, minimum height=1.4cm,  draw=lightgray, fill=lightgray]  at (0,1.3) {};
	\node at (-0.2,1.6) [dot] (a) {}; 
	\node at (-0.2,0.9) [dot] (b) {}; 
	\node at (0.2,1.6) [dot] (c) {}; 
	\node at (0.2,0.9) [dot] (d) {}; 
	\node at (0,1.2) {\scriptsize $n$}; 
	\draw[kernel, bend right=40] (a) to (left); 
	\draw[kernel, bend left=30] (b) to (left); 
	\draw[kernel, bend left=40] (c) to (right); 
	\draw[kernel, bend right=30] (d) to (right); 
	\draw [decorate,decoration={brace,amplitude=7pt}] (-2.6,-1.2) to node[midway, xshift=-0.5cm] {\scriptsize $p$} (-2.6,1.2); 
	\draw [decorate,decoration={brace,amplitude=7pt}] (2.6,1.2) to node[midway, xshift=0.5cm] {\scriptsize $q$} (2.6,-1.2); 
	\end{tikzpicture}, 
\end{equation}
where $\delta_{\tau} = 0$ if $k + \ell$ is even, and $\delta_{\tau} = 1$ if $k+\ell$ is odd. 
The gray cloudy area again stands for a product of cumulant functions over totally $n$ points.
The above graph represents the Wick renormalised object, and the effect of the mass renormalisation $M_{0}^{(\epsilon)}$ is not included, but when $p+q \geq 2$, the graphs will be the same with or without $M_{0}^{(\epsilon)}$. In this section, we focus on the case $p+q \geq 2$ (i.e. ``higher chaos"). 

We first introduce some notations. Let $v^\star$ and $v_{\star}$ denote the starting and end nodes of the barred arrow, respectively, and $v_{\star}$ is connected to the node $0$ by a test function. We let 
\begin{equs}
\pP&=\{\mbox{external (non-contracted noise) vertices connected to } v^\star\}\;, \\
\pP'&=\{\mbox{contracted (noise) vertices connected to } v^\star\}\;.
\end{equs}
We define similarly $\qQ$ and $\qQ'$ except that the vertices are connected to $v_{\star}$, and let $\nN = \pP' \cup \qQ'$. Thus, $\nN$ is the set of all contracted vertices. If we let $p,p',q,q'$ denote the cardinalities of the corresponding sets, then we have 
\begin{equ}
p' + q' = n \;, \qquad p+q+n = k + \ell \;,
\end{equ}
and the total number of external (noise) vertices is $p+q$. 

For simplicity of notations, we will sometimes use the same letter $\nN$ also for edges that connects vertices in $\nN$ with $v^\star$ or $v_{\star}$. We also make similar use of the notations $\pP, \pP', \qQ$ and $\qQ'$. 

Similar as before, in order to make use of the fact that each occurrence of $\eE$ increases the homogeneity by $1$, we need to assign powers of $\epsilon$'s to the edges connecting to the noise vertices (including the contracted ones) to reduce their degrees.\footnote{By reducing degree of the kernel represented by an edge we mean implementing the bound $\eps^\alpha (|x|+\eps)^{-\beta} \lesssim (|x|+\eps)^{-\beta+\alpha}$.}
 If $(k,\ell) = (1,3)$, $(2,2)$ or $(2,3)$, these are standard $\Phi^4_3$ graphs and there will be no powers of $\epsilon$ to assign. In all other cases, there is always positive powers of $\epsilon$ to assign to the edges, and we do it in the following way. 
\begin{enumerate}
	\item Divide the total powers of $\epsilon$'s into $(p+q+n-4-\delta_{\tau})$ pieces, each with power $(\frac{1}{2} - \delta)$ for some sufficiently small $\delta$. We further divide these pieces into two groups such that
	\begin{align*}
	p+q+n-4-\delta_{\tau} = \big(p+p'-2-\1_{\{\ell \phantom{1} \text{odd}\}} \big) + \big( q+q'-1-\1_{\{k \phantom{1} \text{even}\}} \big).  
	\end{align*}
	This always holds since $n=p'+q'$, and $\1_{\{k \phantom{1} \text{even}\}} + \1_{\{\ell \phantom{1} \text{odd}\}} = 1 + \delta_{\tau}$ since we cannot have the situation when $k$ is odd and $\ell$ is even. If $k+\ell-4-\delta_{\tau} \geq 1$, then there will always be positive powers of $\epsilon$'s left. 	
	\item We assign the $\big(p+p'-2-\1_{\{\ell \phantom{1} \text{odd}\}} \big)$ pieces of $(\frac{1}{2}-\delta)$-power of $\epsilon$ to edges in $\pP$ and $\pP'$ in the following way. Assign one piece to each edge in $\pP'$ until using up all $\big(p+p'-2-\1_{\{\ell \phantom{1} \text{odd}\}} \big)$ pieces. If there are still pieces left, we continue assigning one to each of the edges in $\pP$ until finished. 
	
	\item We assign the rest $\big( q+q'-1-\1_{\{k \phantom{1} \text{even}\}} \big)$ to edges in $\qQ \cup \qQ'$ in the same way. 
\end{enumerate}

We then have the following proposition. 

\begin{prop} \label{pr:final_graph}
	The above way of assigning $\epsilon$'s yields the object with the graph representation of the same type as in \eqref{eq:graph} but such that
\begin{enumerate}
	\item Among the $p+p'$ edges from $\pP \cup \pP'$, $2 + \1_{\{\ell \phantom{1} \text{odd}\}}$ of them have degree $a_{e} = 3$, and all the rest have degree $\frac{5}{2} + \delta$. In addition, if there is some edge from $\pP$ that has degree $\frac{5}{2}+\delta$, then all edges from $\pP'$ have degree $\frac{5}{2}+\delta$. 
	
	\item Among the $q+q'$ edges from $\qQ \cup \qQ'$, $1 + \1_{\{k \phantom{1} \text{even}\}}$ of them have degree $3$, and all the rest have degree $\frac{5}{2}+\delta$. In addition, if any edge from $\qQ$ has degree $\frac{5}{2}+\delta$, then all edges from $\qQ'$ have degree $\frac{5}{2}+\delta$. 
	
	\item There is a quantity $\epsilon^{\theta}$ for some $\theta \geq 0$ multiplying the graph, and $\theta = 0$ if and only if $p+q+n-4-\delta_{\tau}=0$. 
\end{enumerate}
\end{prop}
\begin{proof}
	Property $3$ is obvious. Properties $1$ and $2$ come from the fact that we assign $\epsilon$'s to edges in $\pP$ (or $\qQ$) only after all edges in $\pP'$ (or $\qQ'$) are assigned. 
\end{proof}

Note that after the allocation of $\epsilon$'s with the above procedure, the degrees of edges in the graph satisfies
\begin{align*}
\phantom{11} 5|H_{\text{in}} \setminus H_{\star}| + \frac{5}{2} |H_{\text{ex}}| - \sum_{e} a_{e} = -\frac{1}{2} \delta_{\tau} - (p+q+n-4-\delta_{\tau}) \delta, 
\end{align*}
which is at the correct homogeneity (since $|\tau|$ is below $-\frac{1}{2} \delta_{\tau}$, so the above quantity is slightly bigger than $|\tau|$ if $\delta$ is small enough). Similar to the case for first-order objects, in view of Theorem \ref{th:moment-bound}, it now suffices to check that the graph in \eqref{eq:graph} with the properties in Proposition \ref{pr:final_graph} does satisfy Assumption \ref{ass:ele-graph}. 
In the following we still call this graph $(H,\CE)$.

But this time, the graph $H$ is more complicated, and it is hard to check all the conditions in Assumption \ref{ass:ele-graph} for all sub-graphs in a straightforward way. The following lemma gives a simpler procedure in the verification of this assumption. It roughly states that it suffices to check a very small set of sub-graphs $\bar{H}$, and all other sub-graphs will automatically satisfy the assumption if that small set does. 
%
%
To state the lemma we define the following sets of vertices
\[
\pP^\star = \{v^\star\}\cup \pP \;, \quad
\qQ_\star = \{v_\star\}\cup \qQ \;, \quad
\nN_\star^\star = \{v_\star,v^\star\}\cup \nN \;.
\]
\begin{lem} \label{le:simple-verification}
Let $(H,\CE)$ be the labelled graph  defined above.
\begin{itemize}
\item If for every $\bar H \subset H_0$ of cardinality at least 3 such that
\[
\bar H\cap \pP^\star
\in \{\emptyset, \{v^\star\}, \pP^\star  \} \;,\quad
\bar H\cap \qQ_\star
\in  \{\emptyset, \{v_\star\}, \qQ_\star  \} \;, \quad
\bar H\cap \nN^\star_\star
\in \{\emptyset, \nN^\star_\star  \}
\]
satisfies \eqref{e:aEdges}, then item 2 of Assumption~\ref{ass:ele-graph} is satisfied.
\item If for every $\bar H \subset H$ containing $0$ of cardinality at least 2 such that
\[
\bar H\cap \pP^\star
\in \{\emptyset,  \pP^\star  \} \;,\quad
\bar H\cap \qQ_\star
\in  \{\emptyset, \qQ_\star  \} \;, \quad
\bar H\cap \nN^\star_\star
\in \{\emptyset, \nN^\star_\star  \}
\]
satisfies \eqref{e:aEdges1}, then item 3 of Assumption~\ref{ass:ele-graph} is satisfied.
\item If for every non-empty $\bar H \subset H_\star$ such that
\[
\bar H\cap \qQ_\star
=  \emptyset  \;, \quad
\bar H\cap (\pP^\star \cup \nN^\star_\star)
\in \{\emptyset, \pP^\star \cup \nN^\star_\star \}
\]
satisfies \eqref{e:aEdges2}, then item 4 of Assumption~\ref{ass:ele-graph} is satisfied.
\end{itemize}
\end{lem}


\begin{rmk}
	The above lemma should be understood in the following way: in each case, the description of the set $\bar{H}$ given above is the ``worst" case corresponding to the condition, and if the bound is satisfied for that ``worst" case, then it will automatically be satisfied for all other sub-graphs. For example, for condition \eqref{e:aEdges1}, the lemma states that if $v^\star \in \bar{H}$ but $v_{\star}$ not, and the sub-graph $\{v^\star\} \cup \pP$ satisfies bound \eqref{e:aEdges1}, then all other sub-graphs $\bar{H}$ that contains $v^\star$ but not $v_{\star}$ automatically satisfies \eqref{e:aEdges1}. 
	
	Note that in the Item $1$ above (for Condition \eqref{e:aEdges}), we impose further restrictions only for the case when $v^\star$ or $v_{\star}$ is not in $\bar{H}$, but not the case that they are in $\bar{H}$ (unlike Item $2$ for Condition \eqref{e:aEdges1}); this is due to the additional term $\1_{\bar{H}_{\text{ex}} = \emptyset}$ in \eqref{e:aEdges}. 
\end{rmk}
\begin{proof}
	We first claim that if some but not all vertices from $\pP$ (resp. $\qQ$, or one connected component of $\nN$) are in $\bar{H}$, then we can always worsen the bounds for the corresponding conditions by adding to $\bar{H}$ or removing from $\bar{H}$ vertices of $\pP$ (resp. $\qQ$, or one connected component of $\nN$). 
	(By ``worsening" a bound or an inequality, we mean adding or subtracting numbers on both sides so that the difference between the values of the two sides becomes smaller.)
	This will imply that one \textit{never} needs to consider the case when $\bar{H}$ contains some but not all of the vertices from $\pP$, or $\qQ$, or a connected component of $\nN$. 
	
	To see this, we first consider the set $\pP$. 
	Suppose that $\bar H\cap \pP\notin \{\emptyset,  \pP\}$. For \eqref{e:aEdges} and \eqref{e:aEdges1}:
\begin{itemize}
\item
	If $v^\star \in \bar{H}$, then adding one more vertex from $\pP$ into $\bar{H}$ will increase the left hand sides of \eqref{e:aEdges} and \eqref{e:aEdges1} by $3$ or $\frac{5}{2}+\delta$, while increase the right hand sides by only $\frac{5}{2}$. Thus, adding more vertices of $\pP$ into $\bar{H}$ makes both bounds worse. 
\item
If $v^\star \notin \bar{H}$, then removing one vertex in $\pP$ from $\bar{H}$ will not change the left hand sides, but decrease both right hand sides (by at least $\frac{5}{2}$), thus also worsen the bounds.
\end{itemize}
	For the bound \eqref{e:aEdges2}, 
\begin{itemize}
\item
if $v^\star \in \bar{H}$, then adding each other vertex in $\pP$ into $\bar{H}$ does not increase the left hand side, but increases the right hand side by $\frac{5}{2}$; 
\item
if $v^\star \notin \bar{H}$, then removing each vertex in $\pP$ from $\bar{H}$ decreases the left hand side by $3$ or $\frac{5}{2} + \delta$, but decreases the right hand side  by  $\frac{5}{2}$ only. 
\end{itemize}
So in either case the bound becomes worse. 
	This shows that for all the three conditions, the worst case is $\bar H\cap \pP\in \{\emptyset,  \pP\}$.
	The conclusion for $\qQ$ follows from exactly the same argument. 
	
	For the set of contracted vertices $\nN$, let $N$ be one of the  hyper-edges in $\nN$, and
	suppose that $\bar H\cap N\notin \{\emptyset,  N\}$.
	 Then, for bounds \eqref{e:aEdges} and \eqref{e:aEdges1}, removing one vertex of $N$ from $\bar{H}$ decreases both left hand sides by at most $3$, but decreases the right hand sides by $5$, which worsens the bounds. For \eqref{e:aEdges2}, adding additional vertices from $N$ into $\bar{H}$ will increase the left hand side by at most $3$ (depending on whether $v^\star \in \bar{H}$ or not), but will increase the right hand side by $5$, which also worsens the bound. 
	 Therefore, for each hyper-edge in $\nN$, we also only need to consider the whole set instead of part of it. 
	
With the above claimed fact we proceed the proof as follows.

	For \eqref{e:aEdges}, it suffices to show that when $v^\star \notin \bar{H}$
	the bound without $\pP$ is worse than the one with $\pP$.
	Indeed if $v^\star \notin \bar{H}$, then adding all of $\pP$ does not make a difference to the left hand side, but yields an increment of $\frac{5}{2}(p+1)$ on the right hand side. (Note that the increment of the right hand side is $\frac{5}{2} (p+1)$ rather than $\frac{5p}{2}$ due to the additional term $\1_{\bar{H}_{\text{ex}} = \emptyset}$.) Similarly, one can verify that it suffices to check the case $\qQ$ is not in $\bar{H}$ if $v_{\star} \notin \bar{H}$. 
	
	As for noise vertices $\nN$, there are two situations:
\begin{itemize}
\item
	If $\{v^\star,v_\star\} \subset \bar{H}$, then adding $\nN$ into $\bar{H}$ increases the LHS of \eqref{e:aEdges} by
	\begin{equation} \label{eq:increment}
		\sum_{e \in \nN} a_{e} + \frac{5n}{2}, 
	\end{equation}
	where we have used the notation $\nN$ also for edges that connect contracted vertices in $\nN$ with $v^\star$ and $v_{\star}$. Since $|\nN| = n$, and each edge in it has degree at least $\frac{5}{2} + \delta$, \eqref{eq:increment} is clearly larger than the increment  of the RHS (which is $5n$). So the bound with $\nN$ in $\bar H$ implies the one with  $\nN$ not in $\bar H$.
\item
	If either $v^\star$ or $v_{\star}$ is not in $\bar{H}$, then the increment of the RHS of \eqref{e:aEdges} is still $5n$, but that of the LHS takes the form \eqref{eq:increment} with $\sum_{e \in \nN}$ replaced by $\sum_{e \in \tilde\nN}$
	where $\tilde\nN$ is a strict subset of $\nN$ (so $|\tilde\nN|\le n-1$).
	Now, each edge has degree either $\frac{5}{2} + \delta$ or $3$, and there are at most $5$ edges with degree $3$ according to \ref{pr:final_graph},	so the increment of the LHS would be at most 
\begin{equ}
5\cdot 3 +(n-6) \cdot(\frac52+\delta) +\frac{5n}{2}=
5n + (n-6) \delta \;.
\end{equ}
If condition \eqref{e:aEdges} holds with $\nN$ not in $\bar{H}$, then adding a multiple of $\delta$ which can be chosen arbitrarily small to the LHS will not change the validity of the bound, and therefore \eqref{e:aEdges} also holds with all vertices of $\nN$ being in $\bar{H}$.
\end{itemize}
	
	The case for \eqref{e:aEdges1} is essentially the same, except for the noise vertices $\pP$ and $\qQ$ when $v^\star$ or $v_{\star}$ are in $\bar{H}$. If $v^\star \in \bar{H}$, then putting $\pP$ into $\bar{H}$ yields an increment of $\sum_{e \in \pP} a_{e}$ on the left hand side, which is strictly larger than $\frac{5p}{2}$. Thus, the bound is worse if $\pP \subset \bar{H}$. Note that the conclusion here does \text{not} hold for \eqref{e:aEdges} since the increment on the right hand side would be $\frac{5}{2} (p+1)$ instead of $\frac{5p}{2}$. The case for $v_{\star}$ is the same. 
	
	We finally turn to \eqref{e:aEdges2}. Since $v_{\star} \notin \bar{H}$, including $\qQ$ in $\bar{H}$ would yield an increment of $\sum_{e \in \qQ} a_{e}$ on the left hand side, which is strictly bigger than $\frac{5q}{2}$. Since the inequality in this condition is reversed, the bound actually becomes better. So we should only consider the worst case that $\qQ$ is not in $\bar{H}$. One can argue in the same way that we should include $\pP$ in $\bar{H}$ if and only if $v^\star \in \bar{H}$. For contracted vertices $\nN$, if $v^\star \in \bar{H}$, then adding $\nN$ increases the left hand side by $\frac{5n}{2}$ but the right hand side by $5n$. If $v^\star \notin \bar{H}$, then the increment of the left hand side by adding $\nN$ is exactly the quantity \eqref{eq:increment}, which is bigger than $5n$. Thus, we should also include $\nN$ into $\bar{H}$ if and only if $v^\star \in \bar{H}$. 
	
	This completes the proof. 
\end{proof}

\begin{prop}
	If $p+q \geq 2$, then the graph \eqref{eq:graph} satisfies Assumption \ref{ass:ele-graph}. 
\end{prop}
\begin{proof}
	The first condition of Assumption \ref{ass:ele-graph} is trivial as each edge has degree at most $3$. We now give some details in the verification of the other three. 
	
	\begin{enumerate}
		\item Condition \eqref{e:aEdges}. We first consider the case $\{v^\star, v_\star\} \subset \bar{H}$. According to Lemma \ref{le:simple-verification}, we only need to look at the situation when $\nN \subset \bar{H}$. But Item $1$ in that lemma does not give specifications in this case whether to include $\pP$, $\qQ$ or not, so we need to consider all four possibilities for $\pP$ and $\qQ$. If both $\pP$ and $\qQ$ are in $\bar{H}$, then condition \eqref{e:aEdges} reads
		\begin{align*}
			\sum_{e \in \pP \cup \qQ \cup \nN} a_{e} + 3 + \frac{5n}{2} < 5 \big( n + 2 + \frac{1}{2} (p+q-1) \big), 
		\end{align*}
		which certainly holds since $|\pP \cup \qQ \cup \nN| = p + q + n$ and all $a_{e}$ in that sum have degree $\frac{5}{2} + \delta$ except at most $5$ of them which have degree $3$. If neither $\pP$ nor $\qQ$ is in $\bar{H}$, then if we let $r$ denote the number of contracted edges that have degree $3$, the condition reads
		\begin{equation} \label{eq:1stcondition}
			(\frac{5}{2} + \delta)n + (\frac{1}{2} - \delta)r + 3 + \frac{5n}{2} < 5 (n+1). 
		\end{equation}
		Since $p+q \geq 2$, by the assumption on the $\epsilon$-allocation, we necessarily have $r \leq 3$, so the condition holds for all small enough $\delta$. If one of $\pP$ and $\qQ$ is in $\bar{H}$ but the other not, then it is easy to see that the increment of the right hand side of \eqref{eq:1stcondition} is larger than that of the left hand side (the right hand side has an additional increment $\frac{5}{2}$ since $\bar{H}_{\text{ex}} \neq \emptyset$), so Condition \eqref{e:aEdges} also holds. 
		
We now turn to the situation $\{v^\star,v_{\star}\} \cap \bar{H}=\emptyset$. 
In this case, the ``worst" sub-graph $\bar{H}$ according to Lemma \ref{le:simple-verification} is that $\bar{H} = \emptyset$. Note that the bound \eqref{e:aEdges} only requires $|\bar{H}| \geq 3$, and in this case adding any three or more vertices into $\bar{H}$ keeps the left hand side $0$ but yields a positive quantity on the right hand side (except adding all of $\nN$, but then the increments are $\frac{5n}{2} < 5n$, which is still fine). Finally, the case when either $v^\star$ or $v_{\star}$ is in $\bar{H}$ but the other not is easy to verify. Thus, we conclude that the condition \eqref{e:aEdges} is satisfied for all sub-graphs $\bar{H}$ with at least three vertices. 
		
		\item Condition \eqref{e:aEdges1}. 
		The case  $\{v^\star, v_{\star} \}\subset \bar{H}$ is automatically implied by the verified bound \eqref{e:aEdges}, as here the left hand sides for both conditions are the same since neither term involving $r_{e}$ is counted, but the right hand side of \eqref{e:aEdges1} is larger than that of \eqref{e:aEdges}. If neither $v^\star$ nor $v_{\star}$ is in $\bar{H}$, then as long as $\bar{H}$ contains one vertex other than $0$, the right hand side is strictly positive while the left hand side is still $0$, so the bound is also verified. 
		
		If $v^\star \in \bar{H}$ but $v_{\star} \notin \bar{H}$, then by Lemma \ref{le:simple-verification}, we only need to include $\pP$ in $\bar{H}$, so the condition reads
		\begin{align*}
			\sum_{e \in \pP} a_{e} + (3+1-1) < 5(1 + \frac{p}{2}). 
		\end{align*}
		Since the number of edges in $\pP$ of degree $3$ is at most three, the above bound obviously holds. 
		The case $v_{\star} \in \bar{H}$ but $v^\star \notin \bar{H}$ can be checked in the same way. This completes the verification for \eqref{e:aEdges1}. 
		
		\item We finally turn to Condition \ref{e:aEdges2}. Here we assume $v_{\star} \notin \bar{H}$. If $v^\star \notin \bar{H}$, then the worst situation is $\bar{H} = \emptyset$, and we have $0=0$. Since any other case will yield strictly better bound than this one, the bound then holds for any non-empty $\bar{H}$ that does not contain $v^\star$ and $v_{\star}$. If $v^\star \in \bar{H}$, then the worst case is that
		\begin{align*}
			\bar{H} = \{v^\star\} \cup \pP \cup \nN, 
		\end{align*}
		so the condition becomes
		\begin{align*}
			\sum_{e \in \pP \cup \nN} a_{e} + \frac{5n}{2} + 3 + 1 > 5 (n+1 + \frac{p}{2}). 
		\end{align*}
		Again, let $r$ denote the total number of edges in the above sum (in $\pP \cup \nN$) that have degree $3$, then the left hand side above is
		\begin{align*}
			(\frac{5}{2} + \delta)(p+n) + (\frac{1}{2} - \delta) r + \frac{5n}{2} + 4. 
		\end{align*}
		If $\ell \geq 3$, then $r \geq 3$ so the bound always holds. The bound also holds if $\ell = 2$ but $p+n \geq 3$. In the case of the $4$-th homogeneous chaos of $\Psi^{2} \iI (\Psi^{2})$, we have $p+n = r = 2$, so one gets an equality instead of a strict inequality in Condition \eqref{e:aEdges2}. However, we can treat the barred arrow as having degree $a_{3} = 3+\delta$ rather than $3$. This will not violate any of the assumption on our model and give us a strictly inequality in \eqref{e:aEdges2}. 
	\end{enumerate}
	We have now completed the verification for second order objects with chaos order $p+q \geq 2$. 
\end{proof}

\begin{rmk}
	The assumption $p+q \geq 2$ is essential, as Condition \eqref{e:aEdges} is indeed violated when $p+q \leq 1$. For example, for the graph representing the logarithmic divergence in \eqref{eq:log_constant}, if we take $\bar{H} = H_{0}$, then $|\bar{H}_{\text{in}}| = 6$ and $\bar{H}_{\text{ex}} = \emptyset$, so both sides of \eqref{e:aEdges} are $25$, and the \textit{inequality} does not hold. In fact, the same is true for all such graphs with $p+q \leq 1$, and the mass renormalisation $M_{0}$ is needed in order for them to satisfy Assumption \ref{ass:ele-graph}. 
\end{rmk}

\subsection{Bounds for second order objects in $0$-th and $1$-st homogeneous chaos}
We now turn to bounding the objects in the $0$-th and $1$-st order chaos (that is, when $p+q \leq 1$). 
We will need the following lemma taken from Lemma~4.7 in \cite{KPZCLT}.
\begin{lem}  \label{lem:collapse}
Given space-time points $y_1,...,y_n$ and $-5<\alpha_i<0$, one has
\begin{equ} [e:reduce-JM]
\int
\prod_{i=1}^n |y_i-x_i|^{\alpha_i}
|\fC_n^{(\epsilon)} (x_1,...,x_n)|
  \,dx_1...dx_n
\lesssim \eps^{5(n/2-1)}
\int_{\R^5}
 	\prod_{i=1}^n  
	\big(| y_i-x |+\eps\big)^{\alpha_i}
\, dx \,.
\end{equ}
\end{lem}
The $0$-th order objects only occur when $k+\ell$ is even. After renormalisations by $M_{0}^{(\epsilon)}$ (subtraction of the constants $C_{k,\ell}^{(\epsilon)}$'s as defined in \eqref{eq:sum_constants}), they are given by
\begin{equation} \label{eq:zero}
	\epsilon^{\frac{1}{2}(\ell + k - 4)} \bigg(
	\begin{tikzpicture} [scale=0.7,baseline=-0.2cm]
	\node[cloud, cloud puffs=7.7, cloud ignores aspect, minimum width=1.1cm, minimum height=1.2cm,  draw=lightgray, fill=lightgray]  at (0,1.3) {};
	\node at (0.5,-1) [root] (below) {}; 
	\node at (-1.5,0) [dot]  (left) {}; 
	\node at (1.5,0)  [dot]  (right) {}; 
	\node at (-0.8,0.65) {$\cdots$}; 
	\node at (0.9,0.65) {$\cdots$}; 
	\node at (-1,1.5) {\scriptsize $(\ell)$}; 
	\node at (1,1.5) {\scriptsize $(k)$}; 
	\draw[kernel1,bend right =20] (left) to (right); 
	\draw[kernel, bend left=30] (-0.3,1) to (left); 
	\draw[kernel, bend right=40] (-0.3,1.4) to (left); 
	\draw[kernel, bend left = 40] (0.3,1.4) to (right); 
	\draw[kernel, bend right=30] (0.3,1) to (right); 
	\draw[testfcn] (right) to (below); 
	\node at (-0.3,1) [dot] {}; 
	\node at (0.3,1) [dot] {}; 
	\node at (-0.3,1.4) [dot] {}; 
	\node at (0.3,1.4) [dot] {}; 
	\node at (0,1.2) {\scriptsize $\pi$}; 
	\end{tikzpicture}
	-
	\begin{tikzpicture} [scale=0.7,baseline=-0.2cm]
	\node[cloud, cloud puffs=7.7, cloud ignores aspect, minimum width=1.1cm, minimum height=1.2cm,  draw=lightgray, fill=lightgray]  at (0,1.3) {};
	\node at (0.5,-1) [root] (below) {}; 
	\node at (-1.5,0) [dot]  (left) {}; 
	\node at (1.5,0)  [dot]  (right) {}; 
	\node at (-0.8,0.65) {$\cdots$}; 
	\node at (0.9,0.65) {$\cdots$}; 
	\node at (-1,1.5) {\scriptsize $(\ell)$}; 
	\node at (1,1.5) {\scriptsize $(k)$}; 
	\draw[kernel,bend right =20] (left) to (right); 
	\draw[kernel, bend left=30] (-0.3,1) to (left); 
	\draw[kernel, bend right=40] (-0.3,1.4) to (left); 
	\draw[kernel, bend left = 40] (0.3,1.4) to (right); 
	\draw[kernel, bend right=30] (0.3,1) to (right); 
	\draw[testfcn] (right) to (below); 
	\node at (-0.3,1) [dot] {}; 
	\node at (0.3,1) [dot] {}; 
	\node at (-0.3,1.4) [dot] {}; 
	\node at (0.3,1.4) [dot] {}; 
	\node at (0,1.2) {\scriptsize $\pi$}; 
	\end{tikzpicture}
	\bigg) \phantom{1} = \phantom{1} \epsilon^{\frac{1}{2}(\ell+k-4)} \phantom{1} 
	\begin{tikzpicture} [scale=0.7,baseline=-0.2cm]
	\node[cloud, cloud puffs=7.7, cloud ignores aspect, minimum width=1.1cm, minimum height=1.2cm,  draw=lightgray, fill=lightgray]  at (0,1.3) {};
	\node at (0,-1.2) [root] (below) {}; 
	\node at (-1.5,0) [dot]  (left) {}; 
	\node at (1.5,0)  [dot]  (right) {}; 
	\node at (-0.8,0.65) {$\cdots$}; 
	\node at (0.9,0.65) {$\cdots$}; 
	\node at (-1,1.5) {\scriptsize $(\ell)$}; 
	\node at (1,1.5) {\scriptsize $(k)$}; 
	\draw[kernel, bend left=30] (-0.3,1) to (left); 
	\draw[kernel, bend right=40] (-0.3,1.4) to (left); 
	\draw[kernel, bend left = 40] (0.3,1.4) to (right); 
	\draw[kernel, bend right=30] (0.3,1) to (right); 
	\draw[testfcn] (right) to (below); 
	\draw[kernel] (left) to (below); 
    \node at (-0.3,1) [dot] {}; 
    \node at (0.3,1) [dot] {}; 
    \node at (-0.3,1.4) [dot] {}; 
    \node at (0.3,1.4) [dot] {}; 
	\node at (0,1.2) {\scriptsize $\pi$}; 
	\end{tikzpicture}, 
\end{equation}
which represents a deterministic quantity. We have the following bound for it. This in particular implies the first bound of \eqref{eq:main_bound}. 

\begin{prop} \label{pr:zero}
	The quantity on the right hand side of \eqref{eq:zero} is bounded by $\epsilon^{\theta} \lambda^{-\delta}$ for some $\theta \geq 0$ and all sufficiently small $\delta$. Moreover, $\theta = 0$ if and only if $k=\ell=2$ and $\pi$ is the pair-wise contraction. 
\end{prop}
\begin{proof}
	Suppose $\ell + k \geq 6$, or $\ell + k = 4$ but $\pi$ is the partition that contracts all four points together. Then, we have
	\begin{align*}
		\pi = B_{1} \cup \cdots \cup B_{n}
	\end{align*}
	such that either $|B_{j}| \geq 4$ for some $j$ or $\sum_{j}|B_{j}| \geq 6$.
	By consecutively applying Lemma~\ref{lem:collapse} 
	 and  \cite[Lemma~10.14]{Hai14a}\footnote{We actually need a modified version of Lemma~10.14 in \cite{Hai14a}: $\int  (|x-y|+\epsilon)^{-a} (|y|+\epsilon)^{-b} dy \lesssim (|x|+\epsilon)^{-a-b+|s|}$ for all large $a$ and $b$ (no need to assume that $a$ and $b$ are smaller than the dimension).}, 
the object can then be bounded by
	\begin{equation} \label{eq:zero_bound}
		\epsilon^{\frac{1}{2}(\ell+k-4)} \phantom{1}
		\begin{tikzpicture} [scale=0.7,baseline=-0.2cm]
		\node[cloud, cloud puffs=7.7, cloud ignores aspect, minimum width=1.1cm, minimum height=1.3cm,  draw=lightgray, fill=lightgray]  at (0,1.3) {};
		\node at (0,-1) [root] (below) {}; 
		\node at (-1.5,0) [dot]  (left) {}; 
		\node at (1.5,0)  [dot]  (right) {}; 
		\node at (-0.8,0.65) {$\cdots$}; 
		\node at (0.9,0.65) {$\cdots$}; 
		\node at (-1,1.5) {\scriptsize $(\ell)$}; 
		\node at (1,1.5) {\scriptsize $(k)$}; 
		\draw[kernel, bend left=30] (-0.3,1) to (left); 
		\draw[kernel, bend right=40] (-0.3,1.4) to (left); 
		\draw[kernel, bend left = 40] (0.3,1.4) to (right); 
		\draw[kernel, bend right=30] (0.3,1) to (right); 
		\draw[testfcn] (right) to (below); 
		\draw[kernel] (left) to (below); 
		\node at (-0.3,1) [dot] {}; 
		\node at (0.3,1) [dot] {}; 
		\node at (-0.3,1.4) [dot] {}; 
		\node at (0.3,1.4) [dot] {}; 
		\node at (0,1.2) {\scriptsize $\pi$}; 
		\end{tikzpicture}
		\phantom{1} \lesssim \phantom{1} 
		\epsilon^{\frac{1}{2}(\ell + k - 4)} \cdot 
	\Big( \prod_{j=1}^{n} \epsilon^{5(|B_{j}|/2 - 1)} \Big)
		\begin{tikzpicture} [scale=0.7,baseline=-0.3cm]
		\node at (0,-1) [root] (below) {}; 
		\node at (-1.5,0.7) [dot] (left) {}; 
		\node at (1.5,0.7) [dot] (right) {}; 
		\node at (0,1) {\tiny $3(\ell+k) - 5n$}; 
		\draw[generic] (below) to node[labl]{\tiny $3$} (left); 
		\draw[dist] (below) to (right); 
		\draw[gepsilon] (left) to (right); 
		\end{tikzpicture}
		, 
	\end{equation}
	where the dotted line \tikz[baseline=0] \draw[gepsilon] (0,0.15) to (1,0.15); with degree $a_{e}$ denotes a kernel that is bounded by $(|z|+\epsilon)^{-a_{e}}$, where $z$ is the difference between two end points of the line. 
	
	If $\ell + k \geq 6$, then the right hand side above is bounded by
	\begin{align*}
		\epsilon^{\frac{1}{2}(\ell + k - 4)} 
		\begin{tikzpicture} [scale=0.6,baseline=-0.3cm]
		\node at (0,-1) [root] (below) {}; 
		\node at (-1.5,0.7) [dot] (left) {}; 
		\node at (1.5,0.7) [dot] (right) {}; 
		\node at (0,1) {\tiny $\frac{1}{2}(\ell+k)$}; 
		\draw[generic] (below) to node[labl]{\tiny $3$} (left); 
		\draw[dist] (below) to (right); 
		\draw[gepsilon] (left) to (right); 
		\end{tikzpicture}
		\quad \lesssim \quad \epsilon^{\delta}
		\begin{tikzpicture} [scale=0.6,baseline=-0.3cm]
		\node at (0,-1) [root] (below) {}; 
		\node at (-1.5,0.7) [dot] (left) {}; 
		\node at (1.5,0.7) [dot] (right) {}; 
		\node at (0,1) {\tiny $2+\delta$}; 
		\draw[generic] (below) to node[labl]{\tiny $3$} (left); 
		\draw[dist] (below) to (right); 
		\draw[gepsilon] (left) to (right); 
		\end{tikzpicture}
		, 
	\end{align*}
	where the first quantity is obtained by putting the product $\prod \epsilon^{5(|B_{j}|/2-1)}$ into the top edge of the graph and using $\sum_{j} |B_{j}| = \ell + k$, and the bound uses the assumption that $\ell + k  \geq 6$. This gives the convergence to $0$ at the desired topology when $\ell + k \geq 6$. 
	
	If $\ell + k = 4$ but $\pi$ is the single partition that contracts all four points together, then $\pi = B$ with $|B| = 4$. The right hand side of \eqref{eq:zero_bound} could then be reduced to
	\begin{align*}
		\epsilon^{5}
		\begin{tikzpicture} [scale=0.6,baseline=-0.3cm]
		\node at (0,-1) [root] (below) {}; 
		\node at (-1.5,0.7) [dot] (left) {}; 
		\node at (1.5,0.7) [dot] (right) {}; 
		\node at (0,1) {\tiny $7$}; 
		\draw[generic] (below) to node[labl]{\tiny $3$} (left); 
		\draw[dist] (below) to (right); 
		\draw[gepsilon] (left) to (right); 
		\end{tikzpicture}
		\quad \lesssim \quad \epsilon^{\delta}
		\begin{tikzpicture} [scale=0.6,baseline=-0.3cm]
		\node at (0,-1) [root] (below) {}; 
		\node at (-1.5,0.7) [dot] (left) {}; 
		\node at (1.5,0.7) [dot] (right) {}; 
		\node at (0,1) {\tiny $2+\delta$}; 
		\draw[generic] (below) to node[labl]{\tiny $3$} (left); 
		\draw[dist] (below) to (right); 
		\draw[gepsilon] (left) to (right); 
		\end{tikzpicture}
		, 
	\end{align*}
	which also gives the convergence to $0$ at the expected scale. 
	
	The case when $k=\ell=2$ and $\pi$ being the pair-wise contraction is the same, except that no positive power of $\epsilon$ could be created. This completes the proof. 
\end{proof}

We finally turn to the objects in the first chaos. They are obtained from the terms of the form
\begin{equation} \label{eq:kl}
	\begin{tikzpicture} [scale=0.5,baseline=-0.3cm]
	\node at (-1,0) [dot] (left) {}; 
	\node at (1,0) [dot] (right) {}; 
	\node at (-2,1) [dot] (aboveleft) {}; 
	\node at (-2,-1) [dot] (belowleft) {}; 
	\node at (2,-1) [dot] (belowright) {}; 
	\node at (2,1) [dot] (aboveright) {}; 
	\node at (-2,0.2) {$\vdots$}; 
	\node at (2,0.2) {$\vdots$}; 
	\node at (-2.5,0) {\scriptsize $\ell$}; 
	\node at (2.5,0) {\scriptsize $k$}; 
	\draw[kernel1] (left) to (right); 
	\draw[generic] (aboveleft) to (left); 
	\draw[generic] (aboveright) to (right); 
	\draw[generic] (belowleft) to (left); 
	\draw[generic] (belowright) to (right); 
	\end{tikzpicture}
	, 
\end{equation}
multiplied by $\epsilon^{\frac{1}{2}(\ell + k - 5)}$. Here, $\ell$ is odd, $k$ is even, and $\ell \geq 3$, $k \geq 2$. The elements in the first chaos of this object are obtained either from contracting $\ell$ vertices on the left with $k - 1$ vertices on the right, or contracting $\ell-1$ vertices on the left with $k$ vertices on the right. It turns out that we have a similar statement as in the $0$-th chaos: after mass renormalisation, the first chaos of the quantity $(\hPi^{\epsilon}_{0} \tau)(\varphi_{0}^{\lambda})$ has $p$-th moment bounded by $\epsilon^{\theta} \lambda^{(|\tau|+\delta)}$ for $\theta\ge 0$. Furthermore, $\theta = 0$ if and only if $k=2$, $\ell=3$ and $\pi$ is the pair-wise contraction (which could only happen in contraction of the second type). 

To see this, we first consider the contraction of the first type. The object in the first chaos is given by
\begin{equation} \label{eq:first1}
\epsilon^{\frac{1}{2}(\ell+k-5)} \phantom{1}
\begin{tikzpicture} [scale=0.6,baseline=-0.2cm]
\node[cloud, cloud puffs=7.7, cloud ignores aspect, minimum width=1cm, minimum height=1.1cm,  draw=lightgray, fill=lightgray]  at (0,1.4) {};
\node at (0,1.2) {\scriptsize $\pi$}; 
\node at (0.5,-1.5) [root] (below) {}; 
\node at (-1.5,-.2) [dot]  (left) {}; 
\node at (1.5,-.2)  [dot]  (right) {}; 
\node at (-0.8,0.65) {$\cdots$}; 
\node at (0.9,0.65) {$\cdots$}; 
\node at (-1.2,1.3) {\tiny $(\ell)$}; 
\node at (1.5,1.3) {\tiny $(k-1)$}; 
\node at (1.5,-1.5) [var] (noise) {}; 
\draw[kepsilon] (noise) to (right); 
\draw[kernel1,bend right=20] (left) to (right); 
\draw[kernel, bend left=30] (-0.3,1) to (left); 
\draw[kernel, bend right=40] (-0.3,1.4) to (left); 
\draw[kernel, bend left = 40] (0.3,1.4) to (right); 
\draw[kernel, bend right=30] (0.3,1) to (right); 
\draw[testfcn] (right) to (below); 
\node at (-0.3,1) [dot] {}; 
\node at (0.3,1) [dot] {}; 
\node at (-0.3,1.4) [dot] {}; 
\node at (0.3,1.4) [dot] {}; 
\end{tikzpicture}
\phantom{1} - \phantom{1} C_{k-1,\ell,\pi}^{(\epsilon)} \phantom{1}
\begin{tikzpicture} [scale=0.7,baseline=-0.2cm]
\node at (0,0) [dot] (center) {}; 
\node at (0,1.2) [var] (above) {}; 
\node at (0,-1.2) [root] (below) {}; 
\draw[kepsilon] (above) to (center); 
\draw[testfcn] (center) to (below); 
\end{tikzpicture}
\phantom{1} = \phantom{1} - \epsilon^{\frac{1}{2}(\ell+k-5)} \phantom{1}
\begin{tikzpicture} [scale=0.6,baseline=-0.2cm]
\node[cloud, cloud puffs=7.7, cloud ignores aspect, minimum width=1cm, minimum height=1.1cm,  draw=lightgray, fill=lightgray]  at (0,1.4) {};
\node at (0,1.2) {\scriptsize $\pi$}; 
\node at (0,-1.5) [root] (below) {}; 
\node at (-1.5,-0.2) [dot]  (left) {}; 
\node at (1.5,-0.2)  [dot]  (right) {}; 
\node at (-0.8,0.65) {$\cdots$}; 
\node at (0.9,0.65) {$\cdots$}; 
\node at (-1.2,1.3) {\tiny $(\ell)$}; 
\node at (1.5,1.3) {\tiny $(k-1)$}; 
\node at (1.5,-1.5) [var] (noise) {}; 
\draw[kepsilon] (noise) to (right); 
\draw[kernel] (left) to (below); 
\draw[kernel, bend left=30] (-0.3,1) to (left); 
\draw[kernel, bend right=40] (-0.3,1.4) to (left); 
\draw[kernel, bend left = 40] (0.3,1.4) to (right); 
\draw[kernel, bend right=30] (0.3,1) to (right); 
\draw[testfcn] (right) to (below); 
\node at (-0.3,1) [dot] {}; 
\node at (0.3,1) [dot] {}; 
\node at (-0.3,1.4) [dot] {}; 
\node at (0.3,1.4) [dot] {}; 
\end{tikzpicture}
. 
\end{equation}
Again, by applying  Lemma~\ref{lem:collapse} above 
and  \cite[Lemma~10.14]{Hai14a}, it is not hard to check that the sub-graph 
\begin{align*}
\epsilon^{\frac{1}{2}(\ell + k -5)} \phantom{1}
\begin{tikzpicture} [scale=0.8,baseline=0.6cm]
\node[cloud, cloud puffs=7.7, cloud ignores aspect, minimum width=1cm, minimum height=1.1cm,  draw=lightgray, fill=lightgray]  at (0,1.2) {};
\node at (0,1.2) {\scriptsize $\pi$}; 
\node at (-2.4,1.2) [dot]  (left) {}; 
\node at (2.4,1.2)  [dot]  (right) {}; 
\draw[kernel, bend left=60] (-0.3,1) to (left); 
\draw[kernel, bend right=60] (-0.3,1.4) to (left); 
\draw[kernel, bend left = 60] (0.3,1.4) to (right); 
\draw[kernel, bend right=60] (0.3,1) to (right); 
\node at (-0.3,1) [dot] {}; 
\node at (0.3,1) [dot] {}; 
\node at (-0.3,1.4) [dot] {}; 
\node at (0.3,1.4) [dot] {}; 
\node at (-1.5,1.2) {\scriptsize $(\ell$)}; 
\node at (1.5,1.2) {\scriptsize $(k-1)$}; 
\node at (-1,1.3) {$\vdots$}; 
\node at (0.8,1.3) {$\vdots$}; 
\node at (-2.8,1.2) {$v^{\star}$}; 
\node at (2.8,1.2) {$v_{\star}$}; 
\end{tikzpicture}
\end{align*}
represents a kernel $J(v^{\star} - v_{\star})$ such that
\begin{equation} \label{eq:small_kernel}
\| J_{e} \|_{2+\delta; \beta} \lesssim \epsilon^{\theta}
\end{equation}
for all $\delta, \beta > 0$ and some $\theta > 0$ depending on $\delta$. The degree of the kernel $2+\delta$ is relatively easy to check. To see that the norm is bounded by some positive power of $\epsilon$, we first note if $k+\ell \geq 7$, then there is already some $\epsilon$'s multiplying the graph, and one could allocate these powers to the edges in such a way that there is always some positive power left. In the case $k=2$ and $\ell = 3$, then $\pi$ must be the contraction that groups together all three vertices from one side and one from the other side. In this case, one could also create a positive power of $\epsilon$ as in Proposition \ref{pr:zero}. Thus, in all cases, the right hand side of \eqref{eq:first1} could be reduced to the graph
\begin{equation} \label{eq:first}
\begin{tikzpicture} [scale=0.6,baseline=-0.3cm]
\node at (0,-1) [root] (below) {}; 
\node at (-1.5,0.7) [dot] (left) {}; 
\node at (1.5,0.7) [dot] (right) {}; 
\node at (1.5,-1) [var] (noise) {}; 
\node at (0,1) {\tiny $J_{e}$}; 
\draw[generic] (below) to node[labl]{\tiny $3$} (left); 
\draw[dist] (below) to (right); 
\draw[gepsilon] (left) to (right); 
\draw[kepsilon] (noise) to (right); 
\end{tikzpicture}, 
\end{equation}
where the $J_{e}$ satisfies the bound \eqref{eq:small_kernel}. Since its degree $a_{e}$ is $2+\delta$, it is straightforward to check Assumption \ref{ass:ele-graph}. By Theorem \ref{th:moment-bound} and the positive power of $\epsilon$ in front of $\|J_{e}\|_{2+\delta}$, we see that the $L^{p}$ norm of the quantity on the right hand side of \eqref{eq:first1} is bounded by $\epsilon^{\theta} \lambda^{-\frac{1}{2}-\delta}$. Again, we can take $\delta$ small enough so that it vanishes with the correct homogeneity. 

\begin{rmk} \label{rm:representation}
In what follows, we will represent the graph in \eqref{eq:first} and similar graphs by
\begin{equation} \label{eq:first_first}
\epsilon^{\theta} \phantom{1}
\begin{tikzpicture} [scale=0.6,baseline=-0.3cm]
\node at (0,-1) [root] (below) {}; 
\node at (-1.5,0.7) [dot] (left) {}; 
\node at (1.5,0.7) [dot] (right) {}; 
\node at (1.5,-1) [var] (noise) {}; 
\node at (0,1) {\tiny $2+\delta$}; 
\draw[generic] (below) to node[labl]{\tiny $3$} (left); 
\draw[dist] (below) to (right); 
\draw[gepsilon] (left) to (right); 
\draw[kepsilon] (noise) to (right); 
\end{tikzpicture}. 
\end{equation}
The use of this graph here is of course ambiguous as it \textit{does not} suggest that the quantity in \eqref{eq:first1} is bounded by such a graph. In fact, the correct bound is \eqref{eq:first}, which is different from the one above. But for simplicity of notations, we choose to write $\epsilon^{\theta}$ outside the graph and to regard the norm of the upper-edge being $\oO(1)$. 
\end{rmk}

We now turn to the second type of contractions. For those contractions, we have
\begin{equation} \label{eq:first2}
\begin{split}
&\phantom{111} \epsilon^{\frac{1}{2}(\ell+k-5)} \phantom{1}
\begin{tikzpicture} [scale=0.6,baseline=-0.2cm]
\node[cloud, cloud puffs=7.7, cloud ignores aspect, minimum width=1cm, minimum height=1.1cm,  draw=lightgray, fill=lightgray]  at (0,1.3) {};
\node at (0,1.2) {\scriptsize $\pi$}; 
\node at (1.5,-1.5) [root] (below) {}; 
\node at (-1.5,-0.2) [dot]  (left) {}; 
\node at (1.5,-0.2)  [dot]  (right) {}; 
\node at (-0.8,0.65) {$\cdots$}; 
\node at (0.9,0.65) {$\cdots$}; 
\node at (-1.6,1.3) {\tiny $(\ell-1)$}; 
\node at (1.4,1.3) {\tiny $(k)$}; 
\node at (-1.5,-1.5) [var] (noise) {}; 
\draw[kepsilon] (noise) to (left); 
\draw[kernel1, bend right=20] (left) to (right); 
\draw[kernel, bend left=30] (-0.3,1) to (left); 
\draw[kernel, bend right=40] (-0.3,1.4) to (left); 
\draw[kernel, bend left = 40] (0.3,1.4) to (right); 
\draw[kernel, bend right=30] (0.3,1) to (right); 
\draw[testfcn] (right) to (below); 
\node at (-0.3,1) [dot] {}; 
\node at (0.3,1) [dot] {}; 
\node at (-0.3,1.4) [dot] {}; 
\node at (0.3,1.4) [dot] {}; 
\end{tikzpicture}
\phantom{1} - \phantom{1} C_{k,\ell-1,\pi}^{(\epsilon)} \phantom{1}
\begin{tikzpicture} [scale=0.7,baseline=-0.2cm]
\node at (0,0) [dot] (center) {}; 
\node at (0,1.2) [var] (above) {}; 
\node at (0,-1.2) [root] (below) {}; 
\draw[kepsilon] (above) to (center); 
\draw[testfcn] (center) to (below); 
\end{tikzpicture}
\phantom{1} \\
&= \phantom{1} \epsilon^{\frac{1}{2}(\ell+k-5)} \phantom{1} \bigg( \phantom{1}
\begin{tikzpicture} [scale=0.6,baseline=-0.2cm]
\node at (-1,1) [dot] (left) {}; 
\node at (1,1) [dot] (right) {}; 
\node at (-1,-1) [var] (noise) {}; 
\node at (1,-1) [root] (zero) {}; 
\draw[kernelBig] (left) to (right); 
\draw[kepsilon] (noise) to (left); 
\draw[testfcn] (right) to (zero); 
\end{tikzpicture}
\phantom{1} - \phantom{1}
\begin{tikzpicture} [scale=0.6,baseline=-0.2cm]
\node[cloud, cloud puffs=7.7, cloud ignores aspect, minimum width=1cm, minimum height=1.1cm,  draw=lightgray, fill=lightgray]  at (0,1.3) {};
\node at (0,1.2) {\scriptsize $\pi$}; 
\node at (1.5,-1.5) [root] (below) {}; 
\node at (-1.5,-0.2) [dot]  (left) {}; 
\node at (1.5,-0.2)  [dot]  (right) {}; 
\node at (-0.8,0.65) {$\cdots$}; 
\node at (0.9,0.65) {$\cdots$}; 
\node at (-1.6,1.3) {\tiny $(\ell-1)$}; 
\node at (1.4,1.3) {\tiny $(k)$}; 
\node at (-1.5,-1.5) [var] (noise) {}; 
\draw[kepsilon] (noise) to (left); 
\draw[kernel] (left) to (below); 
\draw[kernel, bend left=30] (-0.3,1) to (left); 
\draw[kernel, bend right=40] (-0.3,1.4) to (left); 
\draw[kernel, bend left = 40] (0.3,1.4) to (right); 
\draw[kernel, bend right=30] (0.3,1) to (right); 
\draw[testfcn] (right) to (below); 
\node at (-0.3,1) [dot] {}; 
\node at (0.3,1) [dot] {}; 
\node at (-0.3,1.4) [dot] {}; 
\node at (0.3,1.4) [dot] {}; 
\end{tikzpicture}
\phantom{1} \bigg), 
\end{split}
\end{equation}
where the symbol \tikz \draw[kernelBig] (0,1) to (1,1); denotes the renormalised distribution
\begin{align*}
Q^{(\epsilon)} (\varphi) := \int Q^{(\epsilon)} (z) \big( \varphi(z) - \varphi(0) \big) dz, 
\end{align*}
where the kernel $Q$ is given by
\begin{align*}
Q^{(\epsilon)}(z-z') = Q_{k,\ell-1,\pi}^{(\epsilon)} (z-z') := \phantom{1}
\begin{tikzpicture} [scale=0.8,baseline=-0.0cm]
\node[cloud, cloud puffs=7.7, cloud ignores aspect, minimum width=1cm, minimum height=1.1cm,  draw=lightgray, fill=lightgray]  at (0,1.3) {};
\node at (-1.5,0) [dot]  (left) {}; 
\node at (1.5,0)  [dot]  (right) {}; 
\node at (0,1.2) {\scriptsize $\pi$}; 
\node at (-0.8,0.65) {$\cdots$}; 
\node at (0.9,0.65) {$\cdots$}; 
\node at (-1.6,1.3) {\tiny $(\ell-1)$}; 
\node at (1.4,1.3) {\tiny $(k)$}; 
\node at (-1.5,-0.5) {$z$}; 
\node at (1.5,-0.5) {$z'$}; 
\draw[kernel,bend right=20] (left) to (right); 
\draw[kernel, bend left=30] (-0.3,1) to (left); 
\draw[kernel, bend right=40] (-0.3,1.4) to (left); 
\draw[kernel, bend left = 40] (0.3,1.4) to (right); 
\draw[kernel, bend right=30] (0.3,1) to (right); 
\node at (-0.3,1) [dot] {}; 
\node at (0.3,1) [dot] {}; 
\node at (-0.3,1.4) [dot] {}; 
\node at (0.3,1.4) [dot] {}; 
\end{tikzpicture}
. 
\end{align*}
It is clear that the waved line \tikz[baseline=-0.2] \draw[kernelBig] (0,0.15) to (1,0.15); carries the renormalisation label $r_{e} = -1$. If $\ell + k \geq 7$, or $\ell = 3, k = 2$ but $\pi$ contracts four points altogether, then either the graph itself carries $\epsilon$'s (when $\ell + k \geq 7$), or we can create a positive power of $\epsilon$ using the same argument as in Proposition \ref{pr:zero}. As a consequence, we can assign powers of $\epsilon$'s to the edge \tikz[baseline=-0.2] \draw[kernelBig] (0,0.15) to (1,0.15); so that it has degree $a_{e} = 5 + \delta$. It then follows that the second line of \eqref{eq:first2} is bounded by
\begin{align*}
\epsilon^{\theta} \phantom{1} \bigg( \phantom{1}
\begin{tikzpicture} [scale=0.5,baseline=-0.2cm]
\node at (-1,1) [dot] (left) {}; 
\node at (1,1) [dot] (right) {}; 
\node at (1,-1) [root] (below) {}; 
\node at (-1,-1) [var] (noise) {}; 
\node at (0,1.3) {\tiny $5+\delta,-1$}; 
\draw[generic] (left) to (right); 
\draw[kepsilon] (noise) to (left); 
\draw[testfcn] (right) to (below); 
\end{tikzpicture}
\phantom{1} - \phantom{1} 
\begin{tikzpicture} [scale=0.5,baseline=-0.2cm]
\node at (-1,1) [dot] (left) {}; 
\node at (1,1) [dot] (right) {}; 
\node at (1,-1) [root] (below) {}; 
\node at (-1,-1) [var] (noise) {}; 
\node at (0,1.3) {\tiny $2+\delta$}; 
\draw[generic] (left) to (right); 
\draw[kepsilon] (noise) to (left); 
\draw[generic] (left) to node[labl]{\tiny $3$} (below); 
\draw[testfcn] (right) to (below); 
\end{tikzpicture}
\phantom{1} \bigg)
\end{align*}
for some $\theta > 0$. It is straightforward to check that both graphs satisfy Assumption \ref{ass:ele-graph}, so their $p$-th moments are bounded by $\epsilon^{\theta} \lambda^{-(\frac{1}{2}+\delta) p}$. The case when $\ell = 3, k = 2$ and $\pi$ being the pair-wise contraction is the same except that $\theta = 0$, since there would be no room to create any positive power of $\epsilon$. Finally, we can choose $\delta$ small enough such that $-(\frac{1}{2} + \delta) > |\tau|$ so these bounds imply the first bound in Theorem \ref{th:main_bound}. The second bound for these objects can be obtained in the same way by considering the difference of the kernels $K_{\bar{\epsilon}} - K$ and applying \eqref{eq:kernel_difference}.

\section{Convergence to the limit} \label{sec:limit}

In this section, we collect all the results from the previous sections to identify the limiting process for $u_{\epsilon}$. Recall that our rescaled process $u_{\epsilon}$ satisfies the equation
\begin{equation} \label{eq:identification_eq}
\partial_{t} u_{\epsilon} = \Delta u_{\epsilon} - \sum_{j=0}^{m} \ha_{j}(\theta) \epsilon^{j-1} W_{2j+1, \mu_{\epsilon}}(u_{\epsilon}) + \zeta_{\epsilon}, 
\end{equation}
and we would like to show that $u_{\epsilon}$ converges to the solution of $\Phi^4_3 (\ha_{1})$ family. 

We first give a convergence result for renormalised models, which is the following theorem. 

\begin{thm} \label{th:model_convergence}
Let $\fM_{\epsilon} = M_{\epsilon} \sL_{\epsilon} (\zeta_{\epsilon})$ as built in the previous sections, and $\fM$ be the $\Phi^4_3$ model as built in \cite{Hai14a, HaiXu}. Then, we have $|\!|\!| \fM_{\epsilon}; \fM |\!|\!|_{\epsilon;0} \rightarrow 0$. 
\end{thm}
\begin{proof}
By definition, we have
\begin{equation} \label{eq:model_convergence}
\begin{split}
|\!|\!| \fM_{\epsilon}; \fM |\!|\!|_{\epsilon;0} &= |\!|\!| \fM_{\epsilon}; \fM |\!|\!| + \| \fM_{\epsilon} \|_{\epsilon} \\
&+ \sup_{t \in [0,1]} \| \rR^{\epsilon} \iI(\Xi) (t, \cdot) - \rR \iI(\Xi) (t, \cdot) \|_{\cC^{\eta}}. 
\end{split}
\end{equation}
We first deal the second term on the right hand side above. By the same calculations as in Section 4 of \cite{HaiXu}, we know there exists $\theta > 0$ such that
\begin{equation} \label{eq:group_bound}
\E |\widehat{f}^{\epsilon}_{z}(\sE^{k}_{\ell}\tau)| \lesssim \epsilon^{|\tau| + k - |\ell| + \theta}
\end{equation}
for every $\tau \in \wW$ such that $|\eE^{k}(\tau)| > 0$, and
\begin{equation} \label{eq:positive_bound}
\E |\hPi^{\epsilon}_{z}(\psi_{z}^{\lambda})| \lesssim \lambda^{\beta + \theta} \epsilon^{|\tau| - \beta}, \qquad \beta = \frac{6}{5}
\end{equation}
for all test function $\psi$ that annihilates affine functions. Note that here non-Gaussianity of the noise does not create additional difficulty as we need the first moment only. In addition, both bounds \eqref{eq:group_bound} and \eqref{eq:positive_bound} are uniform over all space-time points $z$ in compact sets and all $\epsilon, \lambda < 1$. By definition of $\|\cdot\|_{\epsilon}$ (\eqref{eq:e_partial_norm}), it follows immediately that
\begin{equation}
\|\fM_{\epsilon}\|_{\epsilon} \rightarrow 0
\end{equation}
in probability. For the third term on the right hand side of \eqref{eq:model_convergence}, we have the desired bound if the models are Gaussian (\cite{HaiXu}). In order to make use of the Gaussian case, we use a diagonal argument as in \cite{KPZCLT}. Recall that for $\epsilon, \bar{\epsilon} \geq 0$, we have set
\begin{align*}
\zeta_{\epsilon, \bar{\eps}} := \zeta_{\epsilon} * \rho_{\bar{\eps}}
\end{align*}
for some space-time mollifier $\rho$ at scale $\bar{\epsilon}$. If $\widehat{\rR}$ denotes the reconstruction operator associated to the limiting $\Phi^4_3$ model (as in \cite{Hai14a} and \cite{HaiXu}) and $\widehat{\rR}^{\epsilon}$ is that associated to the model $M_{\epsilon} \sL_{\epsilon}(\zeta_{\epsilon})$, then by triangle inequality, we have for every $\eta < -\frac{1}{2}$ the bound
\begin{equation} \label{eq:diagonal_noise}
\begin{split}
&\phantom{1} \| \rR^{\epsilon} \iI (\Xi) (t, \cdot) - \rR \iI (\Xi) (t, \cdot) \|_{\cC^{\eta}} \leq \| K*\zeta_{\epsilon}(t, \cdot) - K * \zeta_{\epsilon, \bar{\epsilon}}(t, \cdot) \|_{\cC^{\eta}} \\
&+ \| K * \zeta_{\epsilon, \bar{\epsilon}}(t, \cdot) - K * \zeta_{0, \bar{\epsilon}}(t, \cdot) \|_{\cC^{\eta}} + \| K * \zeta_{0, \bar{\epsilon}}(t, \cdot) - K * \xi (t, \cdot) \|_{\cC^{\eta}}
\end{split}. 
\end{equation}
The last term above only involve Gaussian objects, and is independent of $\epsilon$. Thus, it follows from Proposition 9.5 in \cite{Hai14a} that this term is bounded by $\bar{\epsilon}^{\theta}$ for some positive $\theta$. Also, since $K * \zeta_{\epsilon, \bar{\epsilon}} = \rho_{\bar{\epsilon}} * K * \zeta_{\epsilon}$, it follows from (9.16) in \cite{Hai14a} that
\begin{align*}
\| K*\zeta_{\epsilon}(t, \cdot) - K * \zeta_{\epsilon, \bar{\epsilon}}(t, \cdot) \|_{\cC^{\eta}} \lesssim \bar{\epsilon}^{\theta} \| K * \zeta_{\epsilon} \|_{\xX}, 
\end{align*}
where $\xX = \cC^{\frac{\theta}{2}} (\R, \cC^{\eta + \theta}(\T^{3}))$, and the proportionality constant is independent of $\epsilon < 1$ and $\zeta_{\epsilon}$. By taking $\theta$ small enough, we get
\begin{align*}
\E \| K*\zeta_{\epsilon}(t, \cdot) - K * \zeta_{\epsilon, \bar{\epsilon}}(t, \cdot) \|_{\cC^{\eta}} \lesssim \bar{\epsilon}^{\theta}, 
\end{align*}
uniformly over all $\epsilon < 1$. As for the middle term $\| K * \zeta_{\epsilon, \bar{\epsilon}}(t, \cdot) - K * \zeta_{0, \bar{\epsilon}}(t, \cdot) \|_{\cC^{\eta}}$, for any fixed $\bar{\epsilon}$, the convergence to $0$ in $\cC^{\eta}$ as $\epsilon \rightarrow 0$ is obvious since everything is smooth once $\bar{\epsilon}$ is fixed. Thus, by sending $\epsilon \rightarrow 0$ first, and then $\bar{\epsilon} \rightarrow 0$, we deduce from the above arguments and the form of the right hand side of \eqref{eq:diagonal_noise} that
\begin{align*}
\E \| \rR^{\epsilon} \iI (\Xi) (t, \cdot) - \rR \iI (\Xi) (t, \cdot) \|_{\cC^{\eta}} \rightarrow 0. 
\end{align*}
We finally turn to the first term on the right hand side of \eqref{eq:model_convergence}, $|\!|\!| \fM_{\epsilon}; \fM |\!|\!|$. For this term, we also use the diagonal argument as above and apply the known results for Gaussian models. Let
\begin{align*}
\fM_{\epsilon, \bar{\epsilon}} := M_{\epsilon, \bar{\eps}} \sL_{\epsilon} (\zeta_{\epsilon, \bar{\eps}}), 
\end{align*}
then we have
\begin{equation} \label{eq:diagonal}
\E |\!|\!| \fM_{\epsilon}; \fM |\!|\!| \leq \E |\!|\!| \fM_{\epsilon}; \fM_{\epsilon, \bar{\eps}} |\!|\!| + \E |\!|\!| \fM_{\epsilon,\bar{\eps}}; \fM_{0,\bar{\eps}} |\!|\!| + \E |\!|\!| \fM_{0,\bar{\epsilon}}; \fM |\!|\!|. 
\end{equation}
Note that all of the norms above are usual norms on modelled distributions, and they do not depend on $\epsilon$. The last term is the distance between the Gaussian model $\fM_{0, \bar{\epsilon}}$ and $\fM$, and is bounded by $\bar{\eps}^{\theta}$ for some positive $\theta$ by the convergence results in \cite{HaiXu}. For the first term, by Theorem \ref{th:main_bound} and then the argument for Theorem 10.7 in \cite{Hai14a}, we can show that it is uniformly (in $\epsilon$) bounded by $\bar{\eps}^{\theta}$ for some positive $\theta$. Thus, by sending $\epsilon \rightarrow 0$, \eqref{eq:diagonal} reduces to
\begin{align*}
\limsup_{\epsilon \rightarrow 0} \E |\!|\!| \fM_{\epsilon}; \fM |\!|\!| < C \bar{\epsilon}^{\theta} + \limsup_{\epsilon \rightarrow 0} \E |\!|\!| \fM_{\epsilon,\bar{\eps}}; \fM_{0,\bar{\eps}} |\!|\!|, 
\end{align*}
where $C$ is independent of $\epsilon$. Now, the only remaining term involves two smooth models. For any fixed $\bar{\epsilon}$, by the argument from Section 6 in \cite{KPZCLT}, this term also vanishes as $\epsilon \rightarrow 0$. By further sending $\bar{\epsilon}$ to $0$, we deduce that
\begin{align*}
\E |\!|\!| \fM_{\epsilon}; \fM |\!|\!| \rightarrow 0. 
\end{align*}
This implies the convergence of the solutions to the $\Phi^4_3$ family. 
\end{proof}

We are now ready to state and prove our final result. 

\begin{thm} \label{th:identify_limit}
Let $\phi_{0}^{(\epsilon)} \in \cC^{\gamma,\eta}_{\epsilon}$ such that $\|\phi_{0}^{(\epsilon)}; \phi_{0}\|_{\gamma,\eta;\epsilon} \rightarrow 0$ for some $\phi_{0} \in \cC^{\eta}$. Let $u_{\epsilon}$ solves the PDE \eqref{eq:identification_eq} with initial condition $\phi_{0}^{(\epsilon)}$, where $V_{\theta}$ is a potential satisfying \eqref{eq:pitchfork_v} with. Then, there exists $a < 0$ such that for $\theta = a \epsilon |\log \epsilon| + \oO(\epsilon)$, $u_{\epsilon}$ converges in probability to the solution of $\Phi^4_3 (\ha_{3})$ with initial data $\phi_{0}$, where $\ha_{3} = \frac{\partial^{4} \scal{V}}{\partial x^{4}} (0,0)$, and the convergence takes place in $\cC([0,T], \cC^{\eta}(\T^{3}))$. 
\end{thm}
\begin{proof}
Let $C_{k,\ell}^{(\epsilon)}$'s be the renormalisation constants defined in the previous section, and $M_{\epsilon}$ be the renormalisation map defined with these constants. By Theorem \ref{th:renormalised_eq}, if $\Phi \in \dD^{\gamma,\eta}_{\epsilon}$ solves the fixed point equation \eqref{eq:abstract} and $\widehat{\rR}^{\epsilon}$ is the reconstruction map associated with the renormalised model $\fM_{\epsilon} = M_{\epsilon} \sL_{\epsilon} (\zeta_{\epsilon})$, then the function $v_{\epsilon} = \rR^{\epsilon} \Phi^{(\epsilon)}$ solves the equation
\begin{equation} \label{eq:reconstructed_eq}
\partial_{t} v_{\epsilon} = \Delta v_{\epsilon} - \sum_{j=1}^{m} \lambda_{j}^{(\epsilon)} \epsilon^{j-1} W_{2j+1, \mu_{\epsilon}}(v_{\epsilon}) - \lambda_{0}^{(\epsilon)} v_{\epsilon} - C_{\epsilon} v_{\epsilon} + \zeta_{\epsilon}
\end{equation}
with initial data $\phi_{0}^{(\epsilon)}$, and
\begin{equation} \label{eq:mass_constant}
C_{\epsilon} = \sum_{k, \ell = 1}^{m} \lambda_{k}^{(\epsilon)} \lambda_{\ell}^{(\epsilon)} \bigg( (2k+1)(2k) C_{2k-1,2\ell+1}^{(\epsilon)} + (2k+1)(2\ell+1) C_{2k,2\ell}^{(\epsilon)} \bigg). 
\end{equation}
Comparing \eqref{eq:identification_eq} and \eqref{eq:reconstructed_eq}, we see that as long as we choose
\begin{align*}
\lambda_{j}^{(\epsilon)} = \ha_{j}(\theta) \quad (j \geq 1), \qquad \lambda_{0}^{(\epsilon)} = \epsilon^{-1} \ha_{0}(\theta) - C_{\epsilon}, 
\end{align*}
then the two right hand sides are identical, and we have $u_{\epsilon} = v_{\epsilon} = \widehat{\rR}^{\epsilon} \Phi^{(\epsilon)}$. Also, with the above choice of $\lambda_{j}^{(\epsilon)}$'s, it is straightforward to see that $\lambda_{j}^{(\epsilon)} \rightarrow \ha_{j}$ for each $j \geq 1$ as long as $\theta = \theta(\epsilon) \rightarrow 0$. As for $\lambda_{0}^{(\epsilon)}$, we see from \eqref{eq:mass_constant}, the convergence of $\lambda_{1}^{(\epsilon)}$ to $\ha_{1}$, and the behaviour of $C_{k,\ell}^{(\epsilon)}$'s as in \eqref{eq:log_constant} and \eqref{eq:finite_constant} that we have $C_{\epsilon} = 9 \ha_{1}^{2} C_{\log} |\log \epsilon| + \oO(1)$. Thus, if we take
\begin{align*}
\theta = \frac{9 \ha_{1}^{2} C_{\log}}{\ha_{0}'} \epsilon |\log \epsilon| + \oO(\epsilon),  
\end{align*}
the pre-factors of logarithmic divergences cancel out, so $\lambda_{0}^{(\epsilon)}$ converges to a finite limit. It then follows from Theorems \ref{th:abstract} and \ref{th:model_convergence} that $u_{\epsilon} = \widehat{\rR}^{\epsilon} \Phi^{(\epsilon)}$ converges to the $\Phi^4_3 (\ha_{1})$ family of solutions. 
\end{proof}

\begin{rmk}
	If $\zeta$ is not symmetric, then the object \eqref{eq:kl} (after multiplication of suitable $\epsilon$'s) with $k$ even and $\ell$ odd will have a divergent $0$-th chaos component of order $\epsilon^{-\frac{1}{2}}$. Thus, in order for the renormalised models to converge, one needs to further subtract a constant of the same order. However, such a direct subtraction cannot be attained by merely adjusting the value $\theta$, and must will change the assumption on the potential $V$. 
	
	We expect that in such a case, similar to \cite{HaiXu} where the authors consider Gaussian noise but asymmetric potential, one needs to re-center and rescale the process $u$ (at a scale depending on $\theta$, but in general smaller than $\epsilon^{-1}$) to kill that divergence and obtain either $\Phi^3_3$ or O.U. process in the limit. 
\end{rmk}

\bibliographystyle{Martin}
\bibliography{Refs}

\end{document}